\RequirePackage[final]{graphicx}
\documentclass[a4paper,orivec]{llncs}
\usepackage{afterpage}
\usepackage{yfonts}
\usepackage{mathrsfs}
\usepackage{xcolor}
\usepackage{stmaryrd}
\usepackage{enumitem} 
\usepackage{scalerel}

\usepackage{amsthm}
\usepackage{amsmath}
\usepackage{amsfonts}
\usepackage{mathabx}
\usepackage{nicefrac}
\usepackage{cite}
\usepackage{thm-restate}
\usepackage{pdfsync}
\usepackage{url}
\usepackage{hypernat}
\usepackage{xspace}
\usepackage{comment} 
\usepackage{ifthen}
\usepackage{mathtools}
\mathtoolsset{showonlyrefs} 
\usepackage{proof} 
\usepackage{cmll} \usepackage{hyperref} \hypersetup{final}
\usepackage{tikz}
\usetikzlibrary{arrows,shapes,automata,spy,positioning}

\title{Characteristic Formulae for Session Types (extended version)}

\author{Julien Lange \and Nobuko Yoshida}

\institute{Imperial College London}
\DeclareMathAlphabet{\mathpzc}{OT1}{pzc}{m}{it}

\newcommand{\bigo}[1]{\mathcal{O}(#1)}

\newcommand{\xqedhere}[2]{  \rlap{\hbox to#1{\hfil\llap{\ensuremath{#2}}}}}

\newenvironment{mydef}[1][\unskip]{
  \begin{definition}[#1]
    \rm
    \pushQED{\qed}}  {\popQED\end{definition}}

\newcommand{\bnfsep}{\;\mid\;}
\newcommand{\qst}{\, : \,}
\newcommand{\st}{\, \mid \,}
\newcommand{\defi}{\overset{{\text{def}}}{=}}

\newcommand{\truek}{\top} \newcommand{\falsek}{\bot} 
\newcommand{\VARS}{\mathcal{V}}
\newcommand{\var}[1]{\mathbf{#1}}
\newcommand{\varX}{\var{x}}
\newcommand{\varY}{\var{y}}
\newcommand{\subs}[2]{[\nicefrac{#1}{#2}]}

\newcommand{\compset}[1]{\neg {#1}}

\newcommand{\FORMULA}{\mathcal{F}}
\newcommand{\FORMULAC}{\FORMULA_{\mathit{c}}}
\newcommand{\mmbox}[2]{[ #1 ] #2}
\newcommand{\mmdiamond}[2]{\langle #1 \rangle #2}
\newcommand{\mmnu}[1]{\nu \var{#1} .\,}
\newcommand{\mmnuND}[1]{\nu \var{#1} }

\newcommand{\approxi}[2]{(#1)^{#2}}
\newcommand{\fora}{\phi}
\newcommand{\forb}{\psi}
\newcommand{\ekmodels}{\models_{e,k}}

\newcommand{\msg}[1]{\mathit{#1}}
\newcommand{\Actions}{\mathcal{A}}
\newcommand{\NDActions}{\mathbb{A}}
\newcommand{\any}{{\Op \! a}}
\newcommand{\anydir}[1]{{\Op \msg{#1}}}
\newcommand{\coanydir}[1]{{\dual{\Op} \msg{#1}}}
\newcommand{\snd}[1]{\sendop \! \msg{#1}}
\newcommand{\rcv}[1]{\rcvop \! \msg{#1}}

\newcommand{\TYPES}{\mathcal{T}}

\newcommand{\TYPESC}{\TYPES_{\mathit{c}}}
\DeclareMathOperator{\mysum}{\with}
\newcommand{\inchoicetop}{\oplus} \newcommand{\outchoicetop}{\mysum}
\newcommand{\echoice}{\outchoicetop}
\newcommand{\ichoice}{\inchoicetop}
\newcommand{\inchoiceop}{\bigoplus} \newcommand{\outchoiceop}{\bigwith}
\newcommand{\inchoiceSet}[3]{\inchoiceop_{#1 \in #2} {\snd{\msg{a}}_{#1}} . \, {#3}_{#1}}
\newcommand{\outchoiceSet}[3]{\outchoiceop_{#1 \in #2} {\rcv{\msg{a}}_{#1}} . \, {#3}_{#1}}
\newcommand{\choiceSet}[3]{\bigOp_{#1 \in #2}  {\anydir{\! a}_{#1}} . \, {#3}_{#1}}
\newcommand{\cochoiceSet}[3]{\dual{\bigOp}_{#1 \in #2} \; {\anydir{\!a}_{#1}} . \, {#3}_{#1}}
\newcommand{\choiceSetNoIdx}[3]{\bigOp_{#1 \in #2} \; {\anydir{\!a}_{#1}} . \, {#3}}
\newcommand{\cochoiceSetNoIdx}[3]{\dual{\bigOp}_{#1 \in #2} \; {\anydir{\!a}_{#1}} . \, {#3}}
\newcommand{\cocochoiceSetNoIdx}[3]{\dual{\bigOp}_{#1 \in #2} \; {\coanydir{a}_{#1}} . \, {#3}}
\newcommand{\choice}{\choiceSet{i}{I}{T}} 
\newcommand{\cochoice}{\cochoiceSet{i}{I}{T}} 
\newcommand{\inchoice}{\inchoiceSet{i}{I}{T}} 
\newcommand{\outchoice}{\outchoiceSet{i}{I}{T}}
\newcommand{\rec}[1]{\mathtt{rec}\, \var{#1} . }
\newcommand{\recND}[1]{\mathtt{rec}\, \var{#1} }
\newcommand{\tend}{\mathtt{end}}

\DeclareMathOperator{\spar}{\, | \,}

\newcommand{\inferrule}[1]{\textsc{\scriptsize [#1]}}
\newcommand{\inference}[3]{\infer[\ifthenelse{\equal{#1}{}}{}{\inferrule{#1}}]{#3}{#2}}
\newcommand{\coinference}[3]{\infer=[\ifthenelse{\equal{#1}{}}{}{\inferrule{#1}}]{#3}{#2}}

\DeclareMathOperator{\absubOp}{\unlhd}
\newcommand{\absub}[1]{\absubOp^{{\small #1}}}
\DeclareMathOperator{\subtype}{\leq}
\DeclareMathOperator{\supertype}{\geq}
\newcommand{\ekabsub}{\absub{\bigOp}_{e,k}}
\newcommand{\semarrow}[1]{\xrightarrow{#1}}
\newcommand{\nsemarrow}{\!\!\nrightarrow}
\newcommand{\semarrowC}{\semarrow{}^{\ast}}

\newcommand{\opfun}[1]{
  \ifthenelse{\equal{#1}{\bigcoOp}}
  {\dual{\Op}}
  {\Op}
}

\newcommand{\subform}[1]{\Anyform{#1}{\inchoicetop}} \newcommand{\SUPform}[1] {\Anyform{#1}{\outchoicetop}} \newcommand{\charforname}{\mbox{\boldmath{$F$}}}
\newcommand{\Anyform}[2]{\,{\charforname}(#1,#2)}
\newcommand{\bigtOp}{\maltese} \DeclareMathOperator*{\bigOp}{\scalerel*{\bigtOp}{\Sigma}} \newcommand{\bigcoOp}{\dual{\bigtOp}}
\DeclareMathOperator{\sendop}{!}
\DeclareMathOperator{\rcvop}{?}
\DeclareMathOperator{\Op}{\dagger} \newcommand{\dual}[1]{\overline{#1}}

\newcommand{\nummsg}[1]{\mathit{num}(#1)}
\newcommand{\unfold}[1]{\mathit{unf}(#1)}
\newcommand{\unfoldP}[2]{\mathit{unf}(#1)}
\newcommand{\varocc}[2]{\lvert #1 \rvert_{#2}}

\newcommand{\cons}{\mathfrak{C}}
\newcommand{\choicecons}[2]{#1_{#2}}
\newcommand{\automaton}{\mathcal{M}}
\newcommand{\autoa}{\automaton}
\newcommand{\autob}{\mathcal{N}}
\newcommand{\autof}[1]{\autoa(#1)}
\newcommand{\state}{q}
\newcommand{\States}{Q}
\newcommand{\labs}{\ell}
\newcommand{\edge}{\delta}
\newcommand{\ksub}{\sqsubseteq}
\newcommand{\nksub}{\nsqsubseteq}
\newcommand{\autoprodop}{\blacktriangleleft}
\newcommand{\autoprod}[2]{#1 \autoprodop #2}
\newcommand{\Pstate}{p}
\newcommand{\PStates}{P}

\newcommand{\Pedge}{\Delta}
\newcommand{\PFinals}{F}
\tikzset{
    every state/.style={minimum size=7pt,inner sep=2pt, initial text={},rectangle,rounded corners=3pt},
  mycfsm/.style={
    font=\footnotesize,
    initial where=left,
    ->,>=stealth,auto, node distance=1cm and 1cm,
    scale=1, every node/.style={transform shape},
    baseline=(current  bounding  box.center)
  }
}

\newcommand{\GHsubalgo}{\subtype_{c}}
\newcommand{\GHenvi}{\Gamma}
\newcommand{\judge}[3]{#1 \vdash #2  \GHsubalgo #3}

\newcommand{\myplotsize}{0.482}
\newcommand{\titlelabel}[2]{(#1) #2}

\begin{document}

\maketitle

\begin{abstract}
  Subtyping is a crucial ingredient of session type theory and its
  applications, notably to programming language implementations.
    In this paper, we study effective ways to check whether a session
  type is a subtype of another by applying a characteristic formulae
  approach to the problem.
    Our core contribution is an algorithm to generate a modal
  $\mu$-calculus formula that characterises all the supertypes (or
  subtypes) of a given type.
    Subtyping checks can then be off-loaded to model checkers, thus
  incidentally yielding an efficient algorithm to check safety of
  session types, soundly and completely.
      We have implemented our theory and compared its cost with other
  classical subtyping algorithms.
\end{abstract}

\section{Introduction}
\paragraph{\bf Motivations}
Session types~\cite{THK94,HVK98,betty-survey} have emerged as a
fundamental theory to reason about concurrent programs, whereby not
only the data aspects of programs are typed, but also their
\emph{behaviours} wrt.\ communication.
Recent applications of session types to the
reverse-engineering of large and complex distributed
systems~\cite{zdlc,LTY15} have led to the need of handling potentially
large and complex session types.
Analogously to the current trend of modern compilers to rely on
external tools such as SMT-solvers to solve complex constraints and
offer strong guarantees~\cite{haskell-smt,haskell-measures,LY12,Leino10},
state-of-the-art model checkers can be used to off-load expensive
tasks from session type tools
such as~\cite{scribble,LTY15,YHNN2013}.

A typical use case for session types in software (reverse-)
engineering is to compare the type of an existing program with a
candidate replacement, so to ensure that both are ``compatible''.
In this context, a crucial ingredient of session type theory is the
notion of \emph{subtyping}~\cite{GH99,DemangeonH11,CDY2014} which
plays a key role to guarantee safety of concurrent programs while
allowing for the refinement of specifications and implementations.
Subtyping for session types relates to many classical theories such as
simulations and pre-orders in automata and process algebra theories;
but also to subtyping for recursive types in the
$\lambda$-calculus~\cite{AC93}.
The characteristic formulae
approach~\cite{GS86,si94,steffen89,ails12,ai97,ai07,cs91}, which has
been studied since the late eighties as a method to compute
simulation-like relations in process algebra and automata, appears
then as an evident link between subtyping in session type theory and
model checking theories.
In this paper, we make the first formal connection between session
type and model checking theories, to the best of our knowledge.
We introduce a novel approach to session types subtyping based on
characteristic formulae; and
thus establish that subtyping for session types can be decided in
quadratic time wrt.\ the size of the types.
This improves significantly on the classical algorithm~\cite{GH05}.
Subtyping can then be reduced to a model checking problem and
thus be discharged to powerful model checkers.
Consequently, any advance in model checking technology has an impact
on subtyping.

\paragraph{\bf Example}
Let us illustrate what session types are and what subtyping covers.
Consider a simple protocol between a server and a client, from the
point of view of the server.
The client sends a message of type $\msg{request}$ to the server who
decides whether or not the request can be processed by replying
$\msg{ok}$ or $\msg{ko}$, respectively.
If the request is rejected, the client is offered another chance to
send another request, and so on.
This may be described by the \emph{session type} below
\begin{equation}\label{ex:intro-example-u-1}
  U_1 =
  \rec{\varX}
  \rcv{request} .
  \{ 
  \snd{ok}  . \tend
  \;\,
  \inchoicetop 
  \;
  \snd{ko} . \varX
  \,
  \}
\end{equation}
where $\recND{\varX}$ binds variable $\varX$ in the rest of the
type, 
$\rcv{msg}$ (resp.\ $\snd{msg}$) specifies the reception (resp.\
emission) of a message $\msg{msg}$, $\inchoicetop$ indicates
an \emph{internal choice} between two behaviours, and $\tend$
signifies the termination of the conversation.
An implementation of a server can then be \emph{type-checked} against
$U_1$.

The client's perspective of the protocol may be specified by the
\emph{dual} of $U_1$: \begin{equation}\label{ex:intro-example-u-2}
  \dual{U}_1 = 
  U_2 = 
  \rec{\varX}
  \snd{request} .
  \{ 
  \rcv{ok}  . \tend
  \;\,
  \outchoicetop
  \;
  \rcv{ko} . \varX
  \,
  \}
\end{equation}
where $\outchoicetop$ indicates an \emph{external choice}, i.e., the
client expects two possible behaviours from the server.
A classical result in session type theory essentially says that if the
types of two programs are \emph{dual} of each other, then their
parallel composition is free of errors (e.g., deadlock).

Generally, when we say that $\mathtt{integer}$ is a subtype of
$\mathtt{float}$, we mean that one can safely use an $\mathtt{integer}$
when a $\mathtt{float}$ is required.
Similarly, in session type theory, if $T$ is a \emph{subtype} of a
type $U$ (written $T \subtype U$), then $T$ can be used whenever $U$
is required.
Intuitively, a type $T$ is a \emph{subtype} of a type $U$ if $T$ is
ready to receive no fewer messages than $U$, and $T$ may
not send more messages than $U$~\cite{DemangeonH11,CDY2014}.
For instance, we have
\begin{equation}\label{ex:intro-example-subs}
  \begin{array}{l}
    T_1
    =
    \rcv{request} .  \snd{ok}  . \tend
    \; \subtype  \;
    U_1
    \\
    T_2
    =
    \rec{\varX}
    \snd{request} .
    \{ 
    \rcv{ok}  . \tend
    \,
    \outchoicetop
    \rcv{ko} . \varX
    \,
    \outchoicetop
    \rcv{error} . \tend
    \,
    \}
    \; \subtype  \;
    U_2
  \end{array}
\end{equation}
A server of type $T_1$ can be used whenever a server of type
$U_1$~\eqref{ex:intro-example-u-1} is required ($T_1$ is a more
refined version of $U_1$, which always accepts the request).
A client of type $T_2$ can be used whenever a client of type
$U_2$~\eqref{ex:intro-example-u-2} is required since $T_2$ is a type
that can deal with (strictly) more messages than $U_2$.

In Section~\ref{subsec:CF}, we will see that a session type can be
naturally transformed into a $\mu$-calculus formula that characterises
all its subtypes.
The transformation notably relies on the diamond modality to make some
branches mandatory, and the box modality to allow some branches to be
made optional; see Example~\ref{ex:char-formula}.

\paragraph{\bf Contribution \& synopsis}
In \S~\ref{sec:session-type-theory} we recall session types and give a
new abstract presentation of subtyping.
In \S~\ref{sec:mucal-char-formu} we present a fragment of the
modal $\mu$-calculus and, following~\cite{steffen89}, 
we give a simple algorithm to generate a $\mu$-calculus formula from a
session type that characterises either all its subtypes or all its
supertypes.
In \S~\ref{sec:safety}, 
building on results from~\cite{CDY2014}, we give a sound and
complete model-checking characterisation of safety for session types.
In \S~\ref{sec:algos}, we present two other subtyping algorithms 
for session types:
Gay and Hole's classical algorithm~\cite{GH05} based on inference
rules that unfold types explicitly; and
an adaptation of Kozen et al.'s automata-theoretic
algorithm~\cite{KPS95}.
In \S~\ref{sec:tool},
we evaluate the cost of our approach by comparing its
performances against the two algorithms from \S~\ref{sec:algos}.
Our performance analysis is notably based on a tool that generates
arbitrary well-formed session types.
We conclude and discuss related works in \S~\ref{sec:related}.
Due to lack of space, full proofs are relegated to
Appendix~\ref{app:proofs} (also available online~\cite{appendix}).
Our tool and detailed benchmark results are available
online~\cite{tool}.

\section{Session types and subtyping}\label{sec:session-type-theory}
Session types are abstractions of the behaviour of a program wrt.\ the
communication of this program on a given \emph{session} (or
conversation), through which it interacts with another program (or
component).

\subsection{Session types}\label{sub:session-types}
We use a two-party version of the multiparty session types in~\cite{DY13}.
For the sake of simplicity, we focus on first order session types
(that is, types that carry only simple types (sorts) or values and not
other session types). We discuss how to lift this restriction in
Section~\ref{sec:conc}.
Let $\VARS$ be a countable set of variables (ranged over by $\varX,
\varY$, etc.);
let $\NDActions$ be a (finite) alphabet, ranged over by $a$, $b$,
etc.; and $\Actions$ be the set defined as $\{ \snd{a} \st a \in
\NDActions \} \cup \{ \rcv{a} \st a \in \NDActions \}$.
We let $\Op$ range over elements of $\{ !, ? \}$, so that $\any$
ranges over $\Actions$.
The syntax of session types is given by 
\[
T \coloneqq 
\tend
\bnfsep
\inchoice
\bnfsep
\outchoice
\bnfsep
\rec{\varX} T
\bnfsep
\varX
\]
where $I \neq \emptyset$ is finite,
$a_i \in \NDActions$ for all $i \in I$,
$\msg{a}_i \neq \msg{a}_j$ for $i \neq j$,
and $\varX \in \VARS$.
Type $\tend$ indicates the end of a session.
Type $\inchoice$ specifies an \emph{internal} choice, indicating that
the program chooses to send one of the $\msg{a}_i$ messages, then
behaves as $T_i$.
Type $\outchoice$ specifies an \emph{external} choice, saying that the
program waits to receive one of the $\msg{a}_i$ messages, then
behaves as $T_i$.
Types $\rec{\varX} T$ and $\varX$ are used to specify recursive
behaviours.
We often write, e.g., $\{\snd{a}_1. T_1 \ichoice {\ldots} \ichoice
\snd{a}_k . T_k \}$ for $\inchoiceop_{1 \leq i \leq k} { \snd{a}_i
  . T_i}$, write $\snd{a_1} . T_1$ when $k =1$, similarly for $\outchoice$,
and omit trailing
occurrences of $\tend$.

The sets of free and bound variables of a type $T$ are defined as
usual (the unique binder is the recursion operator $\rec{\varX} T$).
For each type $T$, we assume that two distinct occurrences of a
recursion operator bind different variables, and that no variable has
both free and bound occurrences.
In coinductive definitions, we take an equi-recursive view of types,
not distinguishing between a type $\rec{\varX} T$ and its unfolding $T
\subs{\rec{\varX} T}{\varX}$.
We assume that each type $T$ is
\emph{contractive}~\cite{piercebook02}, e.g., $\rec{\varX} \varX$ is
not a type.
Let $\TYPES$ be the set of all (contractive) session types and
$\TYPESC \subseteq \TYPES$ the set of all closed session types (i.e.,
which do not contain free variables).

\begin{figure}[t]
\[
  \inference{T-out}
  {j \in I}
  {\inchoice \semarrow{\snd{a_j}} T_j }
  \qquad
  \inference{T-in}
  {j \in I}
  {\outchoice \semarrow{\rcv{a_j}} T_j }
  \qquad
  \inference
  {T-rec}
  {T \subs{\rec{\varX} T}{\varX} \semarrow{\any} T'}
  {\rec{\varX} T \semarrow{\any} T'}
  \]
\caption{LTS for session types in $\TYPESC$}\label{fig:lts-types}
\end{figure}
A session type $T \in \TYPESC$ induces a (finite) \emph{labelled
  transition system} (LTS) according to the rules in
Figure~\ref{fig:lts-types}.
We write $T \semarrow{\any}$ if there is $T' \in \TYPES$ such that $T
\semarrow{\any} T'$
and
write $T \nsemarrow$ if $\forall \any \in \Actions \qst \neg (T
\semarrow{\any} )$.
\subsection{Subtyping for session types}
Subtyping for session types was first studied in~\cite{GH99} and
further studied in~\cite{DemangeonH11,CDY2014}.
It is a crucial notion for practical applications of session
types, as it allows for programs to be \emph{refined} while preserving
safety.

We give a definition of subtyping which is parameterised wrt.\
operators $\inchoicetop$ and $\outchoicetop$, so to allow us to
give a common characteristic formula construction for both the subtype
and the supertype relations, cf.\ Section~\ref{subsec:CF}.
Below, we let $\bigtOp$ range over $\{\inchoicetop, \outchoicetop \}$.
When writing $\choice$, we take the convention that $\Op$ refers to
$!$ iff $\bigtOp$ refers to $\inchoicetop$ (and vice-versa for $?$ and
$\outchoicetop$).
We define the (idempotent) duality operator
$\dual{\phantom{\outchoiceop}}$ as follows:
$\dual{\vphantom{\outchoiceop}\inchoicetop} \defi \outchoicetop$,
$\dual{\outchoicetop} \defi \inchoicetop$, $\dual{!} \defi ?$, and $\dual{?} \defi
!$.

\begin{mydef}[Subtyping]\label{def:ab-subtype}
    Fix $\bigtOp \in \{\inchoicetop, \outchoicetop \}$,
  $\absub{\bigOp} \subseteq \TYPESC \times \TYPESC$ is the
  \emph{largest} relation that contains the rules:
  \[
  \resizebox{\textwidth}{!}{$
        \coinference{S-\ensuremath{{\bigtOp}}}
    {
      I \subseteq J
      &
      \forall i \in I \qst T_i \absub{\bigOp}  U_i
    }
    {
      \choice \absub{\bigOp} \choiceSet{j}{J}{U}
    }
        \quad
    \coinference{S-end}
    {}
    {\tend \absub{\bigOp} \tend} 
    \quad
    \coinference{S-\ensuremath{\dual{\bigtOp}}}
    {
      J \subseteq I
      &
      \forall j \in J \qst T_j \absub{\bigOp}  U_j
    }
    {
      \cochoice \absub{\bigOp}    \cochoiceSet{j}{J}{U}
    }
    $}
  \]
  The double line in the rules indicates that the rules
  should be interpreted \emph{coinductively}.
    Recall that we are
  assuming an equi-recursive view of types.
\end{mydef}
We comment Definition~\ref{def:ab-subtype} assuming that $\bigtOp$ is
set to $\inchoicetop$.
Rule $\inferrule{S-\ensuremath{{\bigOp}}}$ says that a type
$\inchoiceSet{j}{J}{U}$ can be replaced by a type that offers no more
messages, e.g., 
$\snd{a} \absub{\,\inchoicetop} \snd{a} \inchoicetop \snd{b}$.
Rule $\inferrule{S-\ensuremath{\dual{\bigOp}}}$ says that a type
$\outchoiceSet{j}{J}{U}$ can be replaced by a type that is ready to
receive at least the same messages, e.g., 
$\rcv{a} \outchoicetop \rcv{b} \absub{\,\inchoicetop} \rcv{a}$.
Rule $\inferrule{S-end}$ is trivial.
It is easy to see that $\absub{\,\inchoicetop} = (
\absub{\,\outchoicetop} )^{-1}$.
In fact, we can recover the subtyping of~\cite{DemangeonH11,CDY2014}
(resp.~\cite{GH99,GH05}) from $\absub{\bigOp}$, by instantiating
$\bigtOp$ to $\inchoicetop$ (resp.\ $\outchoicetop$).
\begin{example}
  Consider the session types from~\eqref{ex:intro-example-subs},
  we have
  $T_1 \absub{\,\inchoicetop} U_1$, $U_1 \absub{\,\outchoicetop} T_1$,
    $T_2 \absub{\,\inchoicetop} U_2$, and $U_2 \absub{\,\outchoicetop} T_2$.
\end{example}
Hereafter, we will write $\subtype$ (resp.\
$\supertype$) for the pre-order $\absub{\,\inchoicetop}$ (resp.\
$\absub{\,\outchoicetop}$).

\section{Characteristic formulae for subtyping}\label{sec:mucal-char-formu}
We give the core construction of this paper:
a function that given a (closed) session type $T$ returns a modal
$\mu$-calculus formula~\cite{Kozen83} that characterises either all
the supertypes of $T$ or all its subtypes.
Technically, we ``translate'' a session type $T$ into a modal
$\mu$-calculus formula $\fora$, so that $\fora$ characterises
all the supertypes of $T$ (resp.\ all its subtypes).
Doing so, checking whether $T$ is a subtype (resp.\ supertype) of $U$
can be reduced to checking whether $U$ is a model of $\fora$, i.e.,
whether $U \models \fora$ holds.

The constructions presented here follow the theory first established
in~\cite{steffen89}; which gives a characteristic formulae approach
for (bi-)simulation-like relations over finite-state processes,
notably for CCS processes.
\subsection{Modal $\mu$-calculus}\label{sub:mucal}
In order to encode subtyping for session types as a model checking
problem it is enough to consider the fragment of the modal $\mu$
calculus below:
\[
\fora \; \coloneqq  \;
\truek 
\bnfsep
\falsek 
\bnfsep
\fora \land \fora 
\bnfsep
\fora \lor \fora 
\bnfsep
\mmbox{\any}{\fora}
\bnfsep
\mmdiamond{\any}{\fora}
\bnfsep
\mmnu{\varX} \fora
\bnfsep
\varX
\]
Modal operators $\mmbox{\any}{}$ and $\mmdiamond{\any}{}$ have
precedence over Boolean binary operators $\land$ and $\lor$; the
greatest fixpoint point operator $\mmnuND{\varX}$ has the lowest
precedence (and its scope extends as far to the right as possible).
Let $\FORMULA$ be the set of all (contractive) modal $\mu$-calculus formulae
and $\FORMULAC \subseteq \FORMULA$ be the set of all closed formulae.
Given a set of actions $A \subseteq \Actions$, we write $\compset{A}$
for $\Actions \setminus A$, and $\mmbox{A}{\fora}$ for $\bigwedge_{\any
  \in A}\mmbox{\any}{\fora}$.

The $n^{th}$ approximation of a fixpoint formula is defined as follows:
\[
\approxi{\mmnu{\varX} \fora}{0}
\; \defi  \;
\truek 
\qquad\qquad\quad
\approxi{\mmnu{\varX} \fora}{n} 
\; \defi \;
\fora \subs{\approxi{\mmnu{\varX} \fora}{n-1}}{\varX} \qquad \text{if } n >0
\]
A \emph{closed} formula $\fora$ is interpreted on the labelled
transition system induced by a session type $T$. 
The satisfaction relation $\models$ between session types and formulae
is inductively defined as follows:
\newcommand{\musemanticdist}{\quad}
\[
\begin{array}
  {
    l@{\musemanticdist \mathit{iff} \musemanticdist}
    l@{\quad}
    l@{\musemanticdist \mathit{iff} \musemanticdist}
    l}
      \multicolumn{2}{l}{T \models \truek
        \vphantom{  T  \semarrow{\any} }
  }
  \\
  T \models \fora_1 \!\land\! \fora_2 
  & T \models \fora_1 \text{ and } T \models \fora_2
  \\
  T \models \fora_1 \!\lor\! \fora_2 & T \models \fora_1 \text{ or } T \models \fora_2
    \vphantom{  T  \semarrow{\any} }
    \\
  T \models \mmbox{\any}{\fora}
  & \forall T' \in \TYPESC \qst \text{if } 
  T  \semarrow{\any} T' \text{ then } T' \models \fora
  \\
  T \models \mmdiamond{\any}{\fora}
  & \exists T' \in \TYPESC \qst
  T  \semarrow{\any} T' \text{ and } T' \models \fora
  \\
  T \models \mmnu{\varX} \fora 
  &
  \forall n \geq 0 \qst T \models \approxi{ \mmnu{\varX} \fora }{n}
    \vphantom{  T  \semarrow{\any} }
  \\
\end{array}
\]
Intuitively, 
$\truek$ holds for every $T$ (while $\falsek$ never holds).
Formula $\fora_1 \land \fora_2$ (resp.\ $\fora_1 \lor \fora_2$) holds if both
components (resp.\ at least one component) of the formula hold in $T$.
The construct $\mmbox{\any}{\fora}$ is a \emph{modal} operator that is satisfied
if for each $\msg{\any}$-derivative $T'$ of $T$, the formula $\fora$ holds in $T'$.
The dual modality is $\mmdiamond{\any}{\fora}$ which holds if there is an
$\msg{\any}$-derivative $T'$ of $T$ such that $\fora$ holds in $T'$.
Construct $\mmnu{\varX} \fora$ is the \emph{greatest} fixpoint
operator (binding $\varX$ in $\fora$).

\subsection{Characteristic formulae}
\label{subsec:CF}

We now construct a $\mu$-calculus formula from a (closed) session
types, parameterised wrt.\ a constructor $\bigtOp$.
This construction is somewhat reminiscent of the \emph{characteristic
  functional} of~\cite{steffen89}.
\begin{mydef}[Characteristic formulae]
  \label{def:char-formula}
    The characteristic formulae of $T \in \TYPESC$ on $\bigtOp$ is given
  by function
    $\charforname : \TYPESC \times \{ \inchoicetop , \outchoicetop \}
  \rightarrow \FORMULAC$, defined as:
    \[
  \Anyform{T}{\bigtOp} \defi 
  \begin{cases}
    \bigwedge_{i \in I} \, \mmdiamond{\opfun{\bigOp}{a_i}}{\Anyform{T_i}{\bigtOp}}
    &
    \text{if } T = \choice
    \\
        \bigwedge_{i \in I} \, \mmbox{\opfun{\bigOp}{a_i}}{\Anyform{T_i}{\bigtOp}}
    &
    \text{if } T = \cochoice
    \\
    \quad
    \land \,
    \bigvee_{i \in I} \, \mmdiamond{\opfun{\bigOp}{a_i}}{\truek}
    \, \land \,
    \mmbox{ \compset{ \{ \opfun{\bigOp}{a_i} \st i \in I\} }}{\falsek}
        \\
        \mmbox{\Actions}{\falsek}
    &
    \text{if } T = \tend
    \\
    \mmnu{\varX} \Anyform{T'}{\bigtOp}
    &
    \text{if } T =  \rec{\varX} T'
    \\
        \varX
    &
    \text{if } T = \varX
                \xqedhere{59pt}{\qed}
          \end{cases}
    \]
  \renewcommand{\qed}{}
\end{mydef}
Given $T \in \TYPESC$, $\Anyform{T}{\ichoice}$ is a $\mu$-calculus
formula that characterises all the \emph{supertypes} of $T$; while
$\Anyform{T}{\echoice}$ characterises all its \emph{subtypes}.
For the sake of clarity, we comment on
Definition~\ref{def:char-formula} assuming that $\bigtOp$ is set to
$\ichoice$.
The first case of the definition makes every branch \emph{mandatory}. 
If $T = \inchoice$, then every internal choice branch that $T$ can
select must also be offered by a supertype, and the relation must hold
after each selection.
The second case makes every branch \emph{optional} but requires at
least one branch to be implemented.
If $T = \outchoice$, then 
($i$) for each of the $\rcv{a_i}$-branch offered by a supertype, the
relation must hold in its $\rcv{a_i}$-derivative,
($ii$) a supertype must offer at least one of the $\rcv{a_i}$
branches, and
($iii$) a supertype cannot offer anything else but the $\rcv{a_i}$
branches.
If $T = \tend$, then a supertype cannot offer any behaviour (recall
that $\falsek$ does not hold for any type).
Recursive types are mapped to greatest fixpoint constructions.

Lemma~\ref{lem:compo} below states the compositionality of the
construction, while Theorem~\ref{thm:main-theorem}, our main result,
reduces subtyping checking to a model checking problem.
A consequence of Theorem~\ref{thm:main-theorem} is that the
characteristic formula of a session type precisely specifies the set
of its subtypes or supertypes.
 
\begin{restatable}{lemma}{lemcompo}
  \label{lem:compo}
  $\Anyform{T \subs{U}{\varX}}{\bigtOp} = \Anyform{T}{\bigtOp} \subs{ \Anyform{U}{\bigtOp} }{\varX}$
\end{restatable}
\begin{proof}
  By structural induction, see Appendix~\ref{proof:lemcompo}.
\end{proof}

\begin{restatable}{theorem}{thmmaintheorem}
  \label{thm:main-theorem}
  $\forall T, U \in \TYPESC  \qst 
  T \absub{\bigOp} U
  \iff
  U \models \Anyform{T}{\bigtOp}$
\end{restatable}
\begin{proof}
  The proof essentially follows the techniques of~\cite{steffen89}, see
  Appendix~\ref{proof:thmmaintheorem}.
\end{proof}

\begin{corollary}\label{cor:TsubU} The following holds:
  \[
  \begin{array}{c@{\qquad\quad}c}
    \begin{array}{l@{\;\,}l}
      (a) & T \subtype U \iff U \models \subform{T}
      \\
      (b)  & U \supertype T \iff T \models \SUPform{U}
    \end{array}
    &
    \begin{array}{l@{\;\,}l}
      (c) & U \models \subform{T}  \iff T \models \SUPform{U}
    \end{array}
  \end{array}
  \]
                  \end{corollary}
\begin{proof}
  By
  Theorem~\ref{thm:main-theorem} and $\subtype =
  \absub{\,\inchoicetop}$, $\supertype = \absub{\,\outchoicetop}$,
  $\subtype = \supertype^{-1}$, and $\absub{\,\inchoicetop} =
  (\absub{\,\outchoicetop} )^{-1}$
\end{proof}
\begin{proposition}\label{prop:char-complexity}
  For all $T, U \in \TYPESC$, deciding whether or not $U \models
  \Anyform{T}{\bigtOp}$ holds can be done in time complexity of
  $\bigo{\lvert T \rvert \times \lvert U \rvert}$, in the worst case;
    where $\lvert T \rvert$ stands for the number of states in the LTS
  induced by $T$.
\end{proposition}
\begin{proof}
  Follows from~\cite{cs91}, since the size of $\Anyform{T}{\bigtOp}$
  increases linearly with $\lvert T \rvert$.
\end{proof}
\begin{example}\label{ex:char-formula}
  Consider session types $T_1$ and $U_1$
  from~\eqref{ex:intro-example-u-1} and~\eqref{ex:intro-example-subs}
    and fix $\Actions = \{ \rcv{request},  \snd{ok},  \snd{ko} \}$.
    Following Definition~\ref{def:char-formula}, we obtain:
        \[
  \begin{array}{rcll}
        \subform{T_1} & = &
    \multicolumn{2}{l}{
      \mmbox{\rcv{request}}{
        \mmdiamond{\snd{ok}}{ \mmbox{\Actions}{\falsek} }
      }
      \;\,
      \land
      \;\,
      \mmdiamond{\rcv{request}}{\truek}
      \;\,
      \land 
      \;\,
      \mmbox{ \neg \{ \rcv{request} \}}{\falsek}
    }
    \\
                                        \SUPform{U_1} & = &
    \mmnu{\varX} 
    \mmdiamond{\rcv{request}}{}     \big(
    &
    \left(
      \mmbox{\snd{ok}}{
        \mmbox{\Actions}{\falsek} 
      }
      \;
      \land
      \;\,
      \mmbox{\snd{ko}}{
        \varX
      }
    \right)
    \\
    &&
    &
    \land
    \;
    \,
    \left(
      \mmdiamond{\snd{ok}}{
        \truek
      }
      \lor
      \mmdiamond{\snd{ko}}{
        \truek
      }
    \right)
    \;\,
    \land 
    \;\,
        \mmbox{ \neg \{ \snd{ok}, \snd{ok} \}}{\falsek}
        \big)
      \end{array}
  \]
    We have $U_1 \models \subform{T_1}$ and $T_1 \models
  \SUPform{U_1}$, as expected (recall tat $T_1 \subtype U_1$).
\end{example}

\section{Safety and duality in session types}\label{sec:safety}
A key ingredient of session type theory is the notion of
\emph{duality} between types.
In this section, we study the relation between duality of session types,
characteristic formulae, and safety (i.e., error freedom).
In particular, building on recent work~\cite{CDY2014} which studies the
preciseness of subtyping for session types, we show how characteristic
formulae can be used to guarantee safety.
\label{subsec:safety}
A system (of session types) is a pair of session types $T$ and $U$
that interact with each other by synchronising over messages. 
We write $T \spar U$ for a system consisting of
$T$ and $U$ and let $S$ range over systems of session types.

\begin{mydef}[Synchronous semantics]\label{def:synch-semantics}
The \emph{synchronous} semantics of a \emph{system} of session types
$T \spar U$ is given by the rule below, in conjunction with the rules
of Figure~\ref{fig:lts-types}.
\[
\inference
{s-com}
{
  T \semarrow{\anydir{a}} T'
  &
  U \semarrow{\coanydir{a}} U'
}
{
  T  \spar U \semarrow{} T' \spar U'
}
\]
We write $\semarrowC$ for the reflexive transitive closure of
$\semarrow{}$.
\end{mydef}
Definition~\ref{def:synch-semantics}
says that two types interact whenever they fire
dual operations.

\begin{example}
  Consider the following execution of system $T_1 \spar U_2$,
  from~\eqref{ex:intro-example-subs}:
    \begin{equation}\label{eq:good-exec}
    \begin{array}{ccll}
                                                                        T_1 \spar U_2
      &
      \; = \;
      &
      \rcv{request} .  \snd{ok}  . \tend
      \; \spar \;
      \rec{\varX}
      \snd{request} . \{ \ldots \}
      \\
      &
      \; \semarrow{\; \; \;} \;
      &
      \snd{ok} . \tend
      \; \spar \;
      \{
      \rcv{ok} . \tend
      \; \echoice \,
      \rcv{ok} . \rec{\varX} \rcv{request} \{ \ldots \}
      \}
                        \; \semarrow{\; \; \;} \;
      &
      \tend \; \spar \; \tend
    \end{array}
  \end{equation}
  \end{example}

\begin{mydef}[Error~\cite{CDY2014} and safety]
  A system $T_1 \spar T_2$ is an \emph{error} if, either:
    \begin{enumerate}[label=(\emph{\alph*})]
  \item \label{en:error-same}
    $T_1 = \choice$ and $T_2 = \choiceSet{j}{J}{U}$, with $\bigtOp$ fixed;
  \item \label{en:error-miss}
    $T_h = \inchoice$ and $T_g = \outchoiceSet{j}{J}{U}$;
    and $\exists i \in I \qst \forall j \in J \qst a_i \neq a_j$, with $h \neq g \in \{1,2\}$; or
  \item \label{en:error-end}
    $T_h = \tend$ and $T_g = \choice$, with $h \neq g \in \{1,2\}$.
  \end{enumerate}
  We say that $S = T \spar U$ is \emph{safe} if for all $S' \qst S
  \semarrowC S'$, $S'$ is not an error.
\end{mydef}
A system of the form~\ref{en:error-same} is an error since both
types are either attempting to send (resp.\ receive) messages.
An error of type~\ref{en:error-miss} indicates that some of the
messages cannot be received by one of the types.
An error of type~\ref{en:error-end} indicates a system where one of
the types has terminated while the other still expects to send
or receive messages.
\begin{mydef}[Duality]\label{def:duality}
  The dual of a formula $\fora \in \FORMULA$, written $\dual{\fora}$
  (resp.\ of a session type $T \in \TYPES$, written $\dual{T}$), is
  defined recursively as follows:
    \[
  \begin{array}{c@{\quad}c}
    \dual{\fora} \defi
    \begin{cases}
      \dual{\fora_1} \land \dual{\fora_2}  & \text{if } \fora =  \fora_1 \land \fora_2  
      \\
      \dual{\fora_1} \lor \dual{\fora_2}   & \text{if } \fora =  \fora_1 \lor \fora_2  
      \\
      \mmbox{\coanydir{a}}{\dual{\fora'}}  & \text{if } \fora =  \mmbox{\any}{\fora'} 
      \\
      \mmdiamond{\coanydir{a}}{\dual{\fora'}}  & \text{if } \fora = \mmdiamond{\any}{\fora'} 
      \\
      \mmnu{\varX} \dual{\fora'}  & \text{if } \fora =  \mmnu{\varX} \fora' 
      \\
      \fora  & \text{if } \fora = \truek, \falsek, \text{ or } \varX
            \xqedhere{196pt}{\qed}
    \end{cases}
    &
    \dual{T} \defi
    \begin{cases}
      \cocochoiceSetNoIdx{i}{I}{\dual{T_i}} & \text{if } T = \choice
      \\
      \rec{\varX} \dual{T'} & \text{if } T = \rec{\varX} T'
      \\
      \varX & \text{if } T = \varX
      \\
      \tend & \text{if } T = \tend
          \end{cases}
      \end{array}
  \]
  \renewcommand{\qed}{}
\end{mydef}
In Definition~\ref{def:duality},
notice that the dual of a formula only rename labels.
\begin{lemma}
  For all $T \in \TYPESC$ and $\fora \in \FORMULAC$,
  $T \models \fora \iff \dual{T} \models \dual{\fora}$.
\end{lemma}
\begin{proof}
  Direct from the definitions of $\dual{T}$ and $\dual{\fora}$ (labels
  are renamed uniformly).
\end{proof}

\begin{restatable}{theorem}{thmcharformduality}
  \label{thm:char-form-duality}
  For all $T \in \TYPES \qst \dual{\Anyform{T}{\bigtOp}} =
  \Anyform{\dual{T}}{\dual{\bigtOp}}$.
\end{restatable}
\begin{proof}
  By structural induction on $T$, see Appendix~\ref{proof:thmcharformduality}.
\end{proof}

Theorem~\ref{thm:safety} follows straightforwardly from~\cite{CDY2014}
and allows us to obtain a sound and complete model-checking based
condition for safety, cf.\ Theorem~\ref{thm:safety-statements}.
\begin{restatable}[Safety]{theorem}{thmsafety}
\label{thm:safety}
$T \spar U$ is safe $\iff$ $(T \subtype \dual{U} \lor \dual{T} \subtype U)$.
\end{restatable}
\begin{proof}
  Direction
  $(\Longrightarrow)$ follows from ~\cite[Table 7]{CDY2014} and
  direction
  $(\Longleftarrow)$ is by coinduction on the derivations of $T
  \subtype \dual{U}$ and $\dual{T} \subtype U$. 
    See Appendix~\ref{app:safety} for details.
\end{proof}
\noindent Finally we achieve: 
\begin{theorem}\label{thm:safety-statements}
  The following statements are equivalent: $\; (a)  {\;\,} T \spar U \text{ is safe }$
  \[
  \begin{array}{l@{\;\,}l@{\qquad\qquad}l@{\;\,}l}
                        (b)
    &
    \dual{U} \models \Anyform{T}{\inchoicetop}
    \lor
    U \models \Anyform{\dual{T}}{\inchoicetop}
    &
    (d)
    &
    U \models \Anyform{\dual{T}}{\outchoicetop}
    \lor
    \dual{U} \models \Anyform{T}{\outchoicetop}
    \\
    (c)
    &
    T \models \Anyform{\dual{U}}{\outchoicetop}
    \lor
    \dual{T} \models \Anyform{{U}}{\outchoicetop}
    &
    (e)
    &
    \dual{T} \models \Anyform{{U}}{\inchoicetop}
    \lor
    {T} \models \Anyform{\dual{U}}{\inchoicetop}
  \end{array}
  \]
\end{theorem}
\begin{proof}
  By direct applications of Theorem~\ref{thm:safety}, then
  Corollary~\ref{cor:TsubU} and
  Theorem~\ref{thm:char-form-duality}.
        \end{proof}

\section{Alternative algorithms for subtyping}\label{sec:algos}
In order to compare the cost of checking the subtyping relation via
characteristic formulae to other approaches, we present two
other algorithms: the original algorithm as given by Gay and
Hole in~\cite{GH05} and an adaptation of Kozen, Palsberg, and
Schwartzbach's algorithm~\cite{KPS95} for recursive subtyping for the
$\lambda$-calculus.
\subsection{Gay and Hole's algorithm}
\label{sub:gay-hole-algo}
\begin{figure}[t]
  \newcommand{\arraytempdist}{0.3cm}
    \[
  \resizebox{\textwidth}{!}{$
    \begin{array}{c}
      \inference{RL}
      {
        \judge{\GHenvi, \rec{\varX} T \GHsubalgo  U  }{T \subs{\rec{\varX} T}{\varX}}{ U}
      }
      {
        \judge{\GHenvi}{\rec{\varX}T}{ U}
      }
            \quad
            \inference{End}
      {
              }
      { 
        \judge{\GHenvi}{\tend}{\tend}
      }
            \quad
            \inference{RR}
      {
        \judge{\GHenvi, T \GHsubalgo  \rec{\varX} U  }{T}{ U \subs{\rec{\varX} U}{\varX}}
      }
      {
        \judge{\GHenvi}{T}{\rec{\varX} U}
      }
            \\[\arraytempdist]
            \inference{Sel}
      {
        I \subseteq J
        &
        \forall i \in I \qst 
                \judge{\GHenvi}{T_i}{U_i}
      }
      {
        \judge{\GHenvi}{\inchoice}{\inchoiceSet{j}{J}{U}}
      }
            \quad
      \inference{Assump}
      {
        T \GHsubalgo U \in \GHenvi
      }
      {
        \judge{\GHenvi}{T}{U}
      }
      \quad
            \inference{Bra}
      {
        J \subseteq I
        &
        \forall j \in J \qst 
                \judge{\GHenvi}{T_j}{U_j}
      }
      {
        \judge{\GHenvi}{\outchoice}{\outchoiceSet{j}{J}{U}}
      }
    \end{array}
    $}
  \]
      \caption{Algorithmic subtyping rules~\cite{GH05}}
  \label{fig:gay-hole-algo}
\end{figure}

 The inference rules of Gay and Hole's algorithm are given in
Figure~\ref{fig:gay-hole-algo} (adapted to our setting).
The rules essentially follow those of Definition~\ref{def:ab-subtype}
but deal explicitly with recursion.
They use judgments $\judge{\GHenvi}{T}{U}$ in which $T$ and $U$ are
(closed) session types and $\GHenvi$ is a sequence of assumed
instances of the subtyping relation, i.e., $\GHenvi = T_1 \GHsubalgo
U_1 , {\scriptstyle  \ldots}, T_k \GHsubalgo U_k$, saying that each pair $T_i
\GHsubalgo U_i$ has been visited.
To guarantee termination, rule $\inferrule{Assump}$
should always be used if it is applicable.

\begin{theorem}[Correspondence{~\cite[Corollary 2]{GH05}}]
  \label{thm:gay-hole-algo}
  $T \subtype U$ if and only if $\judge{\emptyset}{T}{U}$ is derivable
  from the rules in Figure~\ref{fig:gay-hole-algo}.
\end{theorem}

\noindent
Proposition~\ref{prop:gay-complexity}, a contribution of this
paper, states the algorithm's complexity.
\begin{proposition}\label{prop:gay-complexity}
  For all $T, U \in \TYPESC$, the problem of deciding whether or not
  $\judge{\emptyset}{T}{U}$ is derivable has an $\bigo{n^{2^n}}$ time
  complexity, in the worst case; where $n$ is the number of nodes in
  the parsing tree of the $T$ or $U$ (which ever is bigger).
    \end{proposition}
\begin{proof}
  Assume the bigger session type is $T$ and its size is
  $n$ (the number of nodes in its parsing tree).
    Observe that the algorithm in Figure~\ref{fig:gay-hole-algo} needs to
  visit every node of $T$ and relies on explicit unfolding of
  recursive types.
    Given a type of size $n$, its unfolding is of size $\bigo{n^2}$, in the
  worst case.
    Hence, we have a chain $\bigo{n} + \bigo{n^2} + \bigo{n^4} + \ldots$, or
  $\bigo{\sum_{1 \leq i \leq k} n^{2^i}}$,
  where $k$ is a bound on the number of derivations needed for the
  algorithm to terminate.
    According to~\cite[Lemma 10]{GH05}, the number of derivations is
  bounded by the number of sub-terms of $T$, which is $\bigo{n}$.
    Thus, we obtain a worst case time complexity of $\bigo{n^{2^n}}$.
  \end{proof}

 \subsection{Kozen, Palsberg, and Schwartzbach's algorithm}
\label{sub:kozen-et-al}
Considering that the results of~\cite{KPS95} ``\emph{generalise to an
  arbitrary signature of type constructors (\ldots)}'',
we adapt Kozen et al.'s algorithm,
originally designed for subtyping recursive types in the
$\lambda$-calculus. Intuitively, the algorithm reduces the problem of subtyping
to checking the language emptiness of an automaton given by
the product of two (session) types.
The intuition of the theory behind the algorithm is that ``\emph{two
  types are ordered if no common path detects a counterexample}''.
We give the details of our instantiation below.

The set of type constructors over $\Actions$, written
$\cons_\Actions$, is defined as follows:
\[
\cons_\Actions \defi
 \{ \tend \} 
\cup \{ \choicecons{\inchoicetop}{A} \st \emptyset \subset A \subseteq \Actions \}
\cup \{ \choicecons{\outchoicetop}{A} \st \emptyset \subset A \subseteq \Actions \}
\]

\begin{mydef}[Term automata]\label{def:term-automaton}
  A term automaton over $\Actions$ is a tuple
  $
  \automaton = (\States, \, \cons_\Actions, \, \state_0, \, \edge, \, \labs)
  $
  where
          \begin{itemize}
  \item $\States$ is a (finite) set of states,
      \item $\state_0 \in \States$ is the initial state,
      \item $\edge : \States \times \Actions \rightarrow \States$
    is a (partial) function (the \emph{transition function}), and
      \item $\labs : \States \rightarrow \cons_\Actions$ is a (total)
    labelling function
  \end{itemize}
  such that for any $\state \in \States$, 
  if $\labs(\state) \in \{\choicecons{\inchoicetop}{A} , \choicecons{\outchoicetop}{A} \} $, 
  then
  $\edge(\state, \any)$ is defined for all $\any \in A$;
    and 
  for any $\state \in \States$ such that $\labs(\state) = \tend$,
  $\edge(\state, \any)$ is undefined for all $\any \in \Actions$.
    We decorate $\States$, $\edge$, etc.\ with a superscript, e.g.,
  $\automaton$, where necessary.
\end{mydef}

We assume that session types have been ``translated'' to term
automata, the transformation is straightforward (see,~\cite{DY13} for
a similar transformation).
Given a session type $T \in \TYPESC$, we write $\autof{T}$ for its
corresponding term automaton.

\begin{mydef}[Subtyping]\label{def:kozen-sub}
  $\ksub$ is the smallest binary relation on $\cons_\Actions$ such that:
  \[
  \tend \ksub \tend
  \qquad
  \choicecons{\inchoicetop}{A} \ksub \choicecons{\inchoicetop}{B} \iff A \subseteq B
  \qquad
  \choicecons{\outchoicetop}{A} \ksub \choicecons{\outchoicetop}{B} \iff B \subseteq A
  \qedhere
  \]
\end{mydef}
Definition~\ref{def:kozen-sub} essentially maps the rules of
Definition~\ref{def:ab-subtype} to type constructors.
The order $\ksub$ is used in the product automaton
to identify final states, see below.

\begin{mydef}[Product automaton]\label{def:product-automaton}
  Given two term automata $\autoa$ and $\autob$ over $\Actions$, their
  product automaton $\autoprod{\autoa}{\autob}  = (\PStates, \, \Pstate_0, \, \Pedge, \, \PFinals)$
  is such that
          \begin{itemize}
  \item $\PStates = \States^{\autoa} \times \States^{\autob}$ are the states of $\autoprod{\autoa}{\autob}$,
  \item $\Pstate_0 = (\state_0^{\autoa}, \state_0^{\autob})$ is the initial state,
  \item $\Pedge : \PStates \times \Actions  \rightarrow \PStates$ is the partial function which for
    $\state_1 \in \States^{\autoa}$ and $\state_2 \in \States^{\autob}$ gives
    \[
    \Pedge( ( \state_1, \state_2 ), \any ) = ( \edge^{\autoa}(\state_1, \any ) , \edge^{\autob}(\state_2, \any ) )
    \]
  \item $\PFinals \subseteq \PStates$ is the set of \emph{accepting} states:
    $
    \PFinals = \{ \, ( \state_1, \state_2 ) \st  \labs^{\autoa}(\state_1) \nksub \labs^{\autob}(\state_2)  \, \}
    $
  \end{itemize}
Note that $\Pedge( ( \state_1, \state_2 ), \any )$ is defined iff
$\edge^{\autoa}(\state_1, \any)$ and $\edge^{\autob}(\state_2, \any)$ are
defined.
\end{mydef}

Following~\cite{KPS95}, we obtain
Theorem~\ref{def:kozen-correspondence}.
\begin{theorem}\label{def:kozen-correspondence}
  Let $T, U \in \TYPESC$,
  $
  T \subtype U$ iff the language of $\autoprod{\autof{T}}{\autof{U}}$ is empty.
\end{theorem}
Theorem~\ref{def:kozen-correspondence} essentially says that $T
\subtype U$ iff one cannot find a ``common path'' in $T$ and $U$ that
leads to nodes whose labels are not related by $\ksub$, i.e., one
cannot find a counterexample for them \emph{not} being in the
subtyping relation.
\newcommand{\automatadist}{\;\,\,}
\begin{example}
    Below we show the constructions for $T_1$~\eqref{ex:intro-example-u-1} 
  and $U_1$~\eqref{ex:intro-example-subs}.
  \[
  \begin{array}{c@{\automatadist}c@{\automatadist}c@{\automatadist}c}
    \begin{tikzpicture}[mycfsm]
      \node[state, fill=gray!15] (q0) {$\choicecons{\outchoicetop}{\{ \rcv{request} \}}$};
      \node[state, below of=q0] (q1) {$\choicecons{\inchoicetop}{ \{ \snd{ok} \} }$};
      \node[state, below of=q1] (q2) {$\tend$};
            \node[below of=q2,yshift=0.5cm,font=\tiny] {$\autof{T_1}$};
            \path
      (q0) edge node [right] {$\rcv{request}$} (q1)
      (q1) edge node [left] {$\snd{ok}$} (q2)
      ;
    \end{tikzpicture}
    &
    \begin{tikzpicture}[mycfsm]
      \node[state, fill=gray!15] (q0) {$\choicecons{\outchoicetop}{\{ \rcv{request} \}}$};
      \node[state, below of=q0] (q1) {$\choicecons{\inchoicetop}{ \{ \snd{ok} , \snd{ko} \} }$};
      \node[state, below of=q1] (q2) {$\tend$};
            \node[below of=q2,yshift=0.5cm,font=\tiny] {$\autof{U_1}$};
            \path
      (q0) edge [bend right] node [left] {$\rcv{request}$} (q1)
      (q1) edge node [left] {$\snd{ok}$} (q2)
      (q1) edge [bend right] node [right] {$\snd{ko}$} (q0)
      ;
    \end{tikzpicture}
    &
    \begin{tikzpicture}[mycfsm]
      \node[state, fill=gray!15] (q0) {$
        \choicecons{\outchoicetop}{\{ \rcv{request} \}}
        \ksub
        \choicecons{\outchoicetop}{\{ \rcv{request} \}}
        $};
      \node[state, below of=q0] (q1) {$
        \choicecons{\inchoicetop}{ \{ \snd{ok} \} }
        \ksub
        \choicecons{\inchoicetop}{ \{ \snd{ok} , \snd{ko} \} }
        $};
      \node[state, below of=q1] (q2) {$
        \tend
        \ksub
        \tend
        $};
            \node[below of=q2,yshift=0.5cm,font=\tiny] {$\autoprod{\autof{T_1}}{\autof{U_1}}$ };
            \path
      (q0) edge node [left] {$\rcv{request}$} (q1)
      (q1) edge node [left] {$\snd{ok}$} (q2)
            ;
    \end{tikzpicture}
    &
    \begin{tikzpicture}[mycfsm]
      \node[state, fill=gray!15] (q0) {$
        \choicecons{\outchoicetop}{\{ \rcv{request} \}}
        \ksub
        \choicecons{\outchoicetop}{\{ \rcv{request} \}}
        $};
      \node[state, below of=q0, accepting] (q1) {$
        \choicecons{\inchoicetop}{ \{ \snd{ok} , \snd{ko} \} }
        \nksub
        \choicecons{\inchoicetop}{ \{ \snd{ok} \} }
        $};
      \node[state, below of=q1] (q2) {$
        \tend
        \ksub
        \tend
        $};
            \node[below of=q2,yshift=0.5cm,font=\tiny] {$ \autoprod{\autof{U_1}}{\autof{T_1}}$};
            \path
      (q0) edge node [left] {$\rcv{request}$} (q1)
      (q1) edge node [left] {$\snd{ok}$} (q2)
            ;
    \end{tikzpicture}
      \end{array}
  \]
  Where initial states are shaded and accepting states are denoted by
  a double line.
    Note that the language of $\autoprod{\autof{T_1}}{\autof{U_1}}$ is
  empty (no accepting states).
\end{example}

 \begin{proposition}\label{prop:kozen-complexity}
  For all $T, U \in \TYPESC$, the problem of deciding whether or not
  the language of $\autoprod{\autof{T}}{\autof{U}}$ is empty has a
  worst case complexity of $\bigo{\lvert T \rvert \times \lvert U
    \rvert}$;
      where $\lvert T \rvert$ stands for the number of states in
  the term automaton $\autof{T}$.
    \end{proposition}
\begin{proof}
  Follows from the fact that the algorithm in~\cite{KPS95} has a
  complexity of $\bigo{n^2}$, see~\cite[Theorem 18]{KPS95}.
    This complexity result applies also to our instantiation, assuming that
  checking membership of $\ksub$ is relatively
  inexpensive, i.e., $\lvert A \rvert \ll \lvert
  \States^\automaton \rvert$ for each $\state$ such that
  $\labs^\automaton(\state) \in \{\choicecons{\inchoicetop}{A} ,
  \choicecons{\outchoicetop}{A} \}$.
  \end{proof}

\section{Experimental evaluation}
\label{sec:tool}
Proposition~\ref{prop:gay-complexity} states that Gay and Hole's
classical algorithm has an exponential complexity; while the other
approaches have a quadratic complexity
(Propositions~\ref{prop:char-complexity}
and~\ref{prop:kozen-complexity}).
The rest of this section presents several experiments that give a
better perspective of the \emph{practical} cost of these approaches.
\begin{figure}[t]
  \centering
    \begin{tikzpicture}
    [node distance = -0.2cm and 0cm, spy using outlines={rectangle, magnification=2.5, size=4.8cm, connect spies}]
            \node (randomlin) at (0,0)
    {\includegraphics[width=\myplotsize\textwidth]{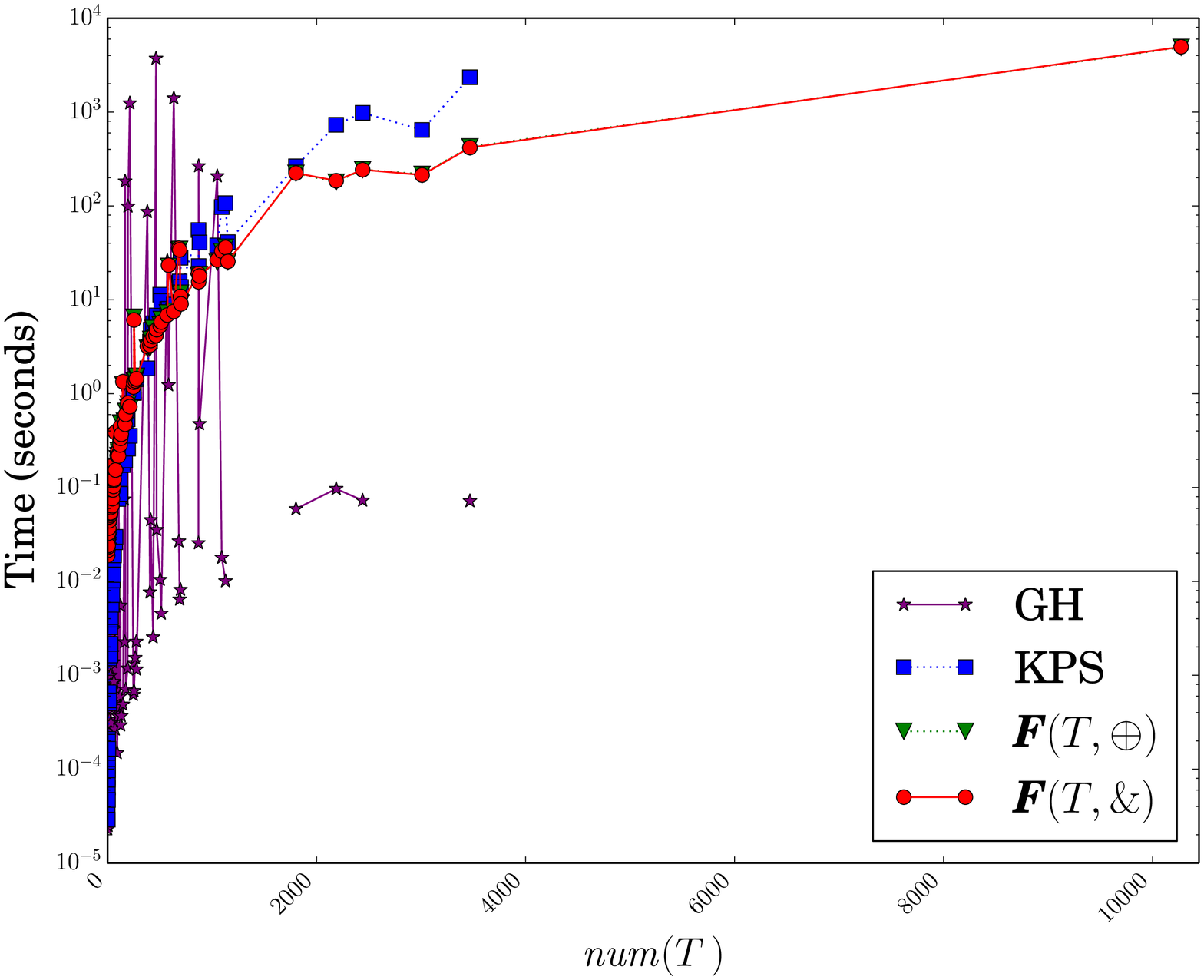}};
        \node[align=center, below = of randomlin,yshift=0cm]
    (nodea)
    {\titlelabel{a}{Arbitrary session types (lin.\ scale)}};
        \node[right = of randomlin] (randomlog) 
    {\includegraphics[width=\myplotsize\textwidth]{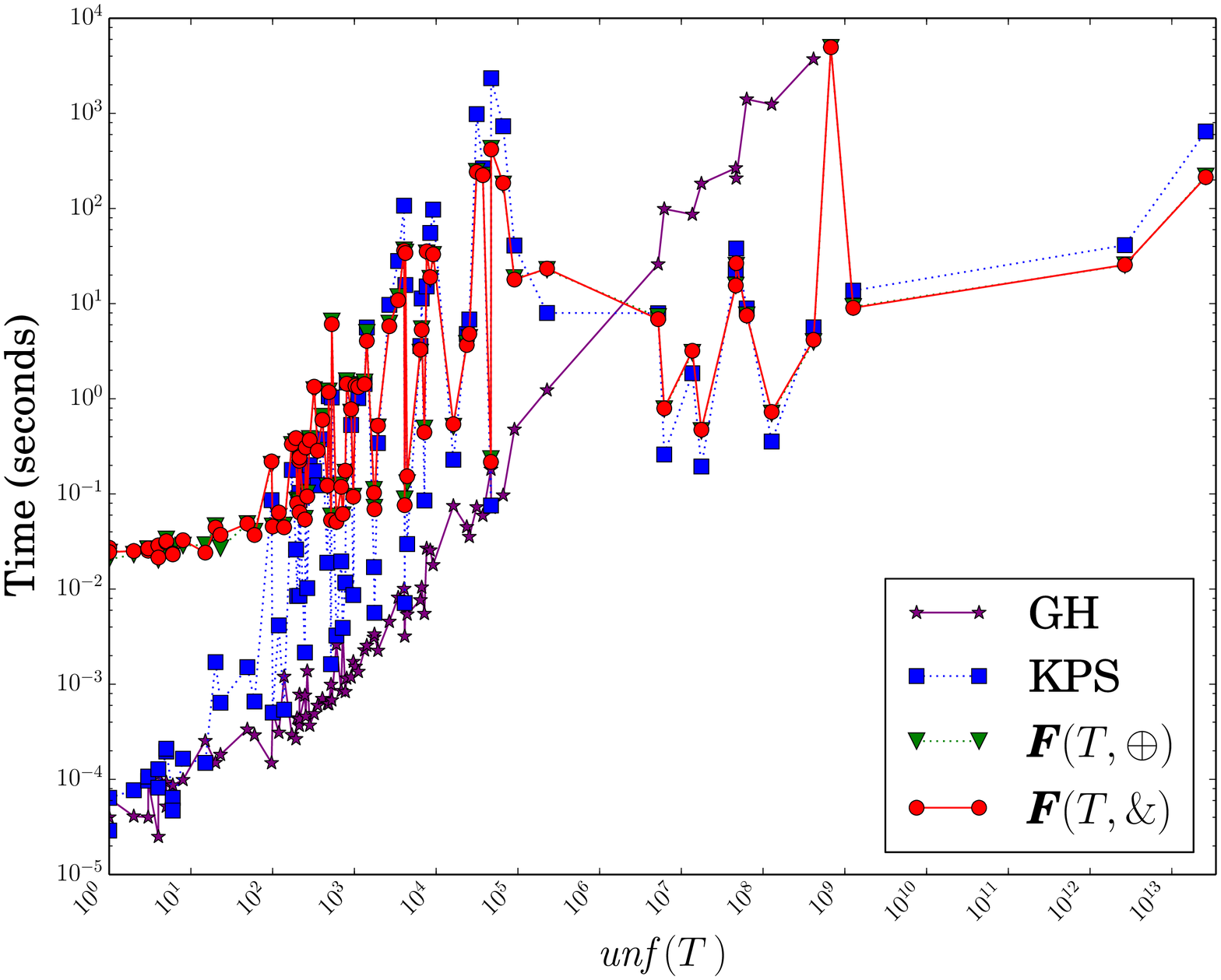}};
        \node[align=center, below = of randomlog,yshift=0cm]
    {\titlelabel{b}{Arbitrary session types (log.\ scale)}};
        \node[below = of nodea] (norec) 
    {\includegraphics[width=\myplotsize\textwidth]{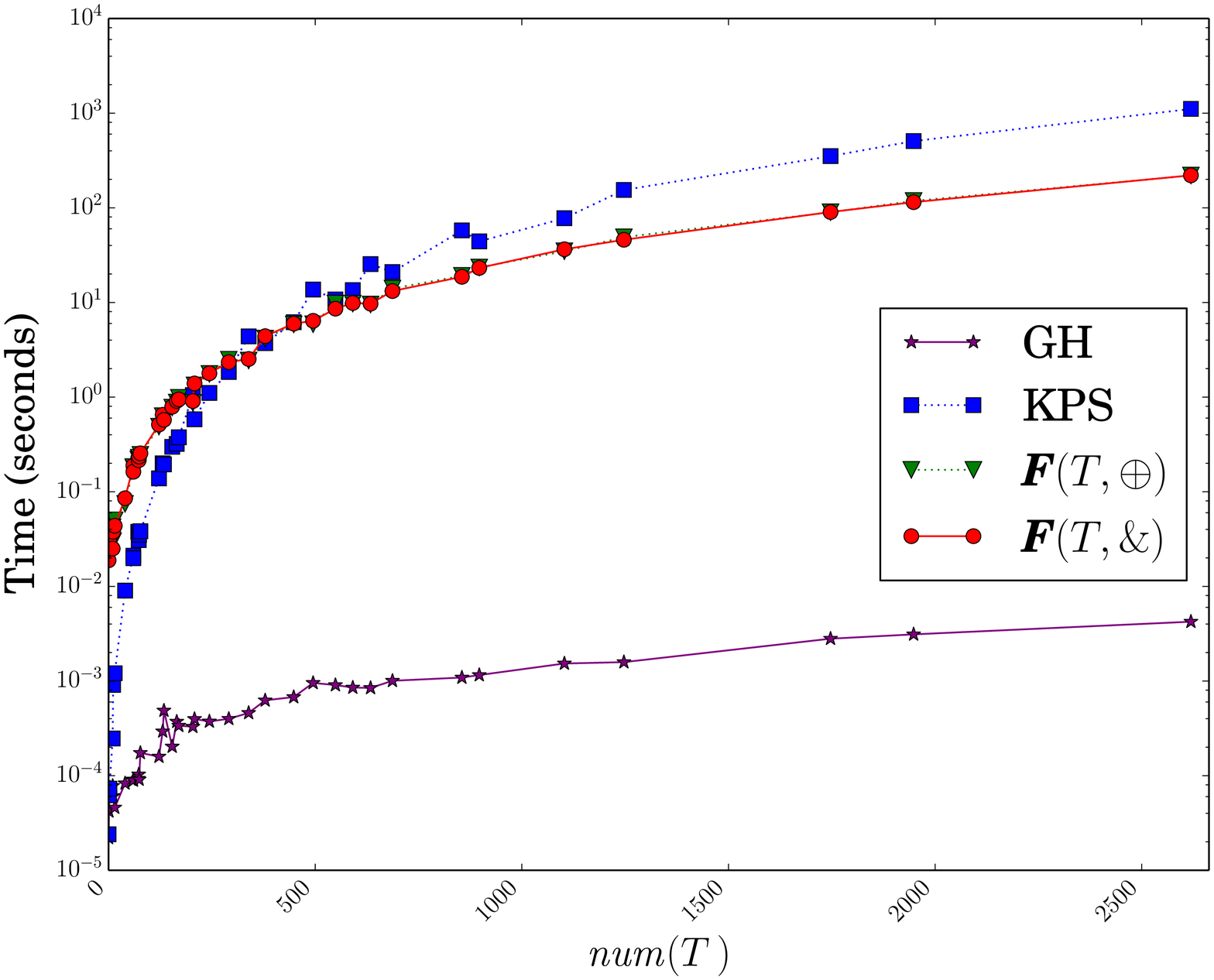}};
        \node[align=center, below = of norec,yshift=0cm]
    {\titlelabel{c}{No recursive definition (lin.\ scale)}};
      \end{tikzpicture}
    \caption{Benchmarks (1)}
  \label{fig:plots-1}
\end{figure}
 \subsection{Implementation overview and metrics}
We have implemented three different approaches to checking whether two
given session types are in the subtyping relation given in
Definition~\ref{def:ab-subtype}.
The tool~\cite{tool}, written in Haskell, consists of three main parts:
($i$) A module that translates session types to the mCRL2 specification
language~\cite{groote2014modeling} and generates a characteristic
($\mu$-calculus) formula (cf.\ Definition~\ref{def:char-formula}),
respectively;
($ii$) A module implementing the algorithm of~\cite{GH05} (see
Section~\ref{sub:gay-hole-algo}), which relies on the Haskell
$\mathtt{bound}$ library to make session types unfolding
as efficient as possible.
($iii$) A module implementing our adaptation of Kozen et
al.'s algorithm~\cite{KPS95}, see Section~\ref{sub:kozen-et-al}.
Additionally, we have developed an accessory tool which generates
arbitrary session types using Haskell's QuickCheck
library~\cite{quickcheck}.

The tool invokes the mCRL2 toolset~\cite{mcrl2} (release version {\tt
  201409.1}) to check the validity of a $\mu$-calculus formula on a
given model.
We experimented invoking mCRL2 with several parameters and
concluded that the default parameters
gave us the best performance overall.
Following discussions with mCRL2 developers, we
have notably experimented with a parameter that pre-processes the
$\mu$-calculus formula to ``insert dummy fixpoints in modal
operators''.
This parameter gave us better performances in some cases, but dramatic
losses for ``super-recursive'' session types.
Instead, an addition of ``dummy fixpoints'' while generating
the characteristic formulae gave us the best results
overall.\footnote{This optimisation was first suggested on the mCRL2
  mailing list.}
The tool is thus based on a slight modification of
Definition~\ref{def:char-formula} where a modal operator
$\mmbox{\any}{\fora}$ becomes $\mmbox{\any}{\mmnu{\var{t}}\fora}$
(with $\var{t}$ fresh and unused) and similarly for
$\mmdiamond{\any}{\fora}$.
Note that this modification does not change the semantics of the
generated formulae.

We use the following functions to measure the size of a session type.
\[
\resizebox{\textwidth}{!}{$
  \begin{array}{ll}
    \nummsg{T}  \, \defi 
    &
    \unfoldP{T}{X}  \, \defi  
    \\
    \quad
    \left\{
      \!
      \begin{array}{ll}
        0 
        & \text{if } T = \tend \; \text{or} \; T = \varX 
        \\
        \nummsg{T'} 
        &  \text{if } T = \rec{x} T'
        \\
        \lvert I \rvert \!+\!   \sum_{i \in I} \nummsg{T_i}
        & \text{if } T = {\choice}
      \end{array}
    \right.
    \;
    &
    \quad
    \left\{
      \!
      \begin{array}{ll}
        0 
        & \text{if } T = \tend \; \text{or} \; T = \varX
        \\
        (1 \!+\! \varocc{T'}{\varX}) \!\times\! \unfold{T'} 
        &  \text{if } T = \rec{x} T'
        \\
        \lvert I \rvert \!+\!   \sum_{i \in I} \unfold{T_i}
        & \text{if } T = {\choice} 
      \end{array}
    \right.
      \end{array}
  $}
\]
Function $\nummsg{T}$ returns the \emph{number of messages} in $T$.
Letting $\varocc{T}{\varX}$ be the number of times variable $\varX$
appears \emph{free} in session type $T$, function $\unfold{T}$ returns
the number of messages in the unfolding of $T$.
Function $\unfold{T}$ takes into account the structure of a type wrt.\
recursive definitions and calls (by unfolding once every recursion
variable).
\begin{figure}[t]
  \centering
    \begin{tikzpicture}
    [node distance = -0.2cm and -0cm, spy using outlines={rectangle, magnification=2.5, size=4.5cm, connect spies}] 
        \node (recsend) at (0,0)
    {\includegraphics[width=\myplotsize\textwidth]{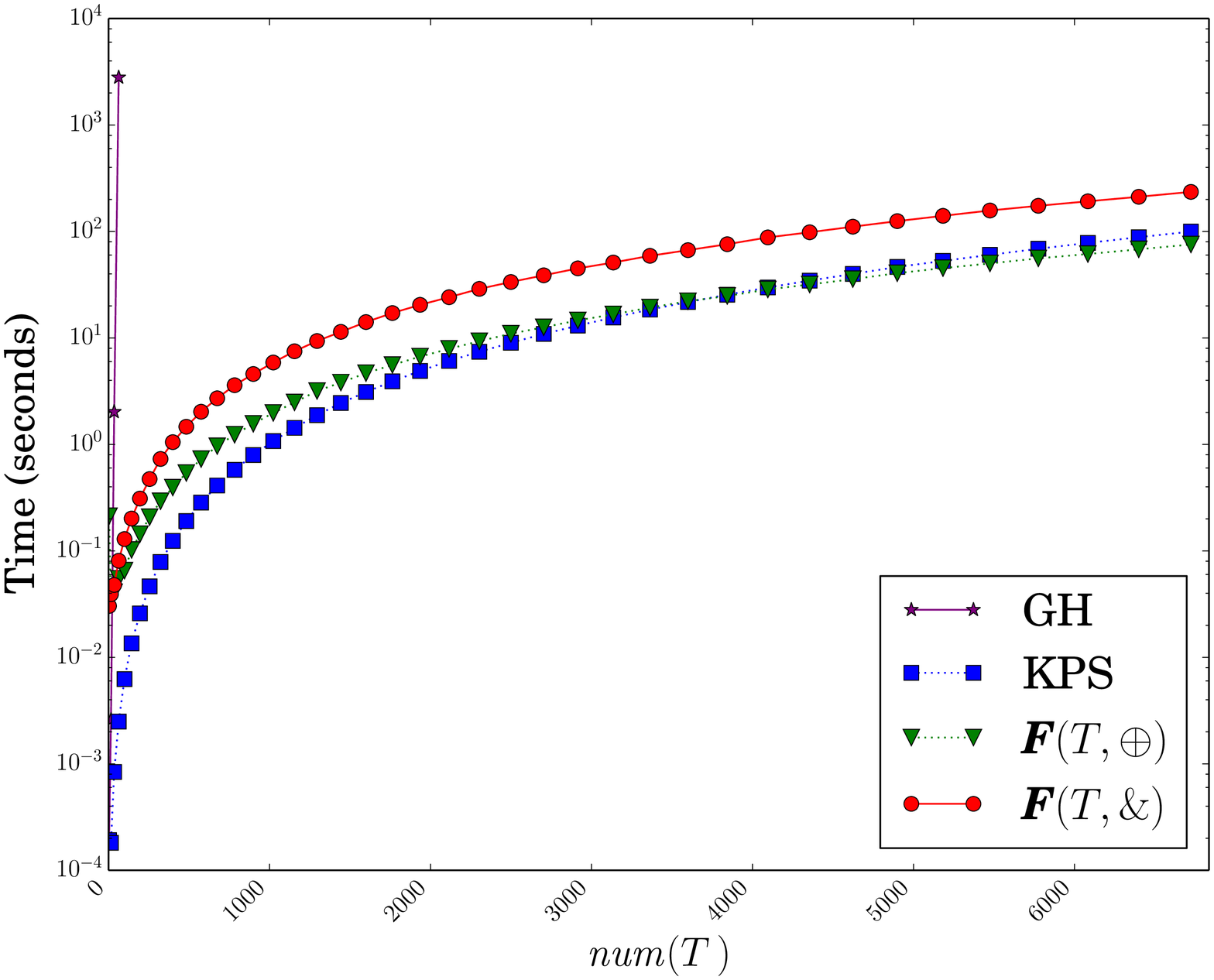}};
        \node[align=center, below = of recsend,yshift=0cm]
    (noded)
    {\titlelabel{d}{Recursive-Send (lin.\ scale)}};
            \node[below = of noded] (unfolded) 
    {\includegraphics[width=\myplotsize\textwidth]{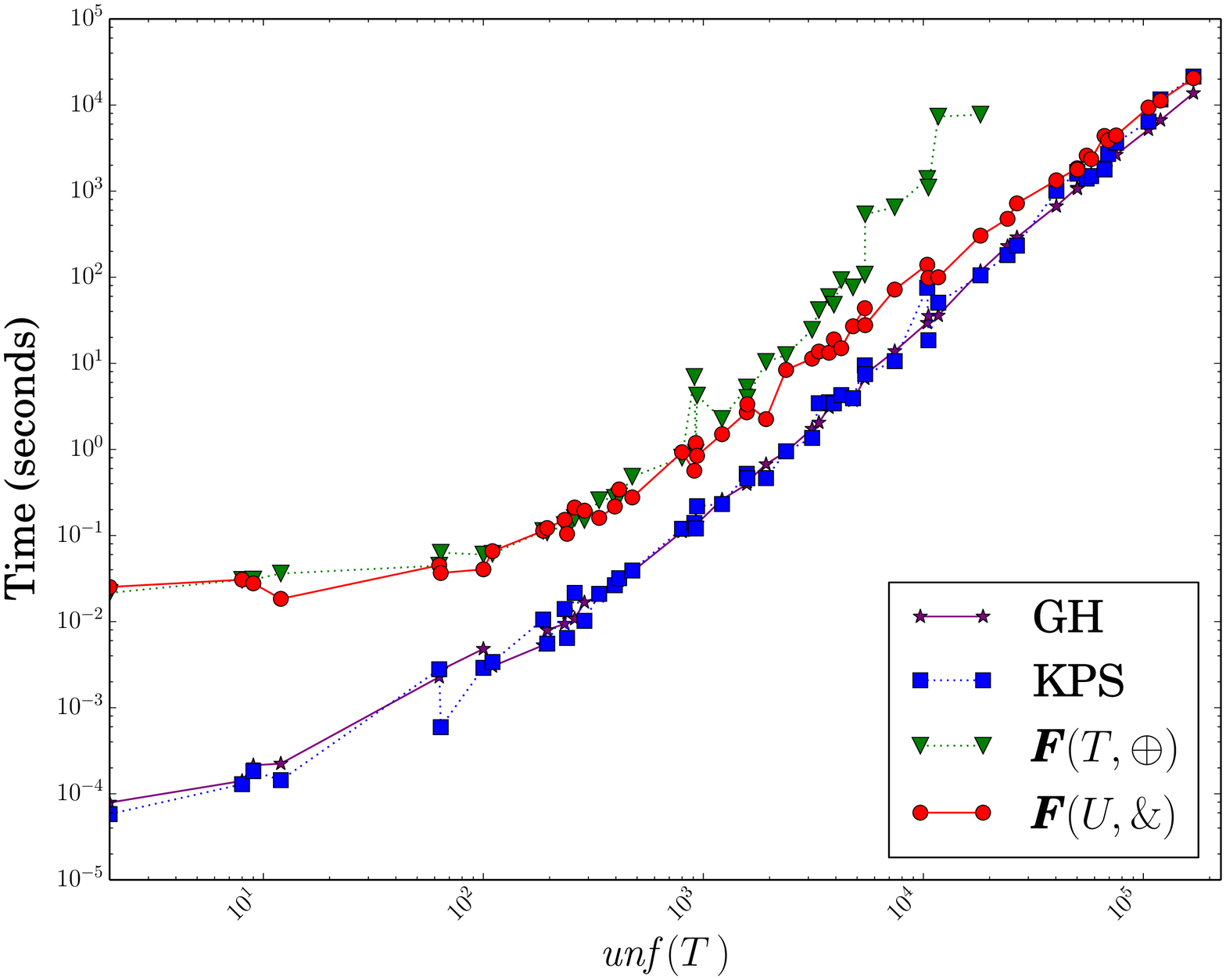}};
        \node[align=center, below = of unfolded,yshift=0cm]
    {\titlelabel{f}{One-time unfolded (log.\ scale)}};
        \node[left = of unfolded] (recrcv) 
    {\includegraphics[width=\myplotsize\textwidth]{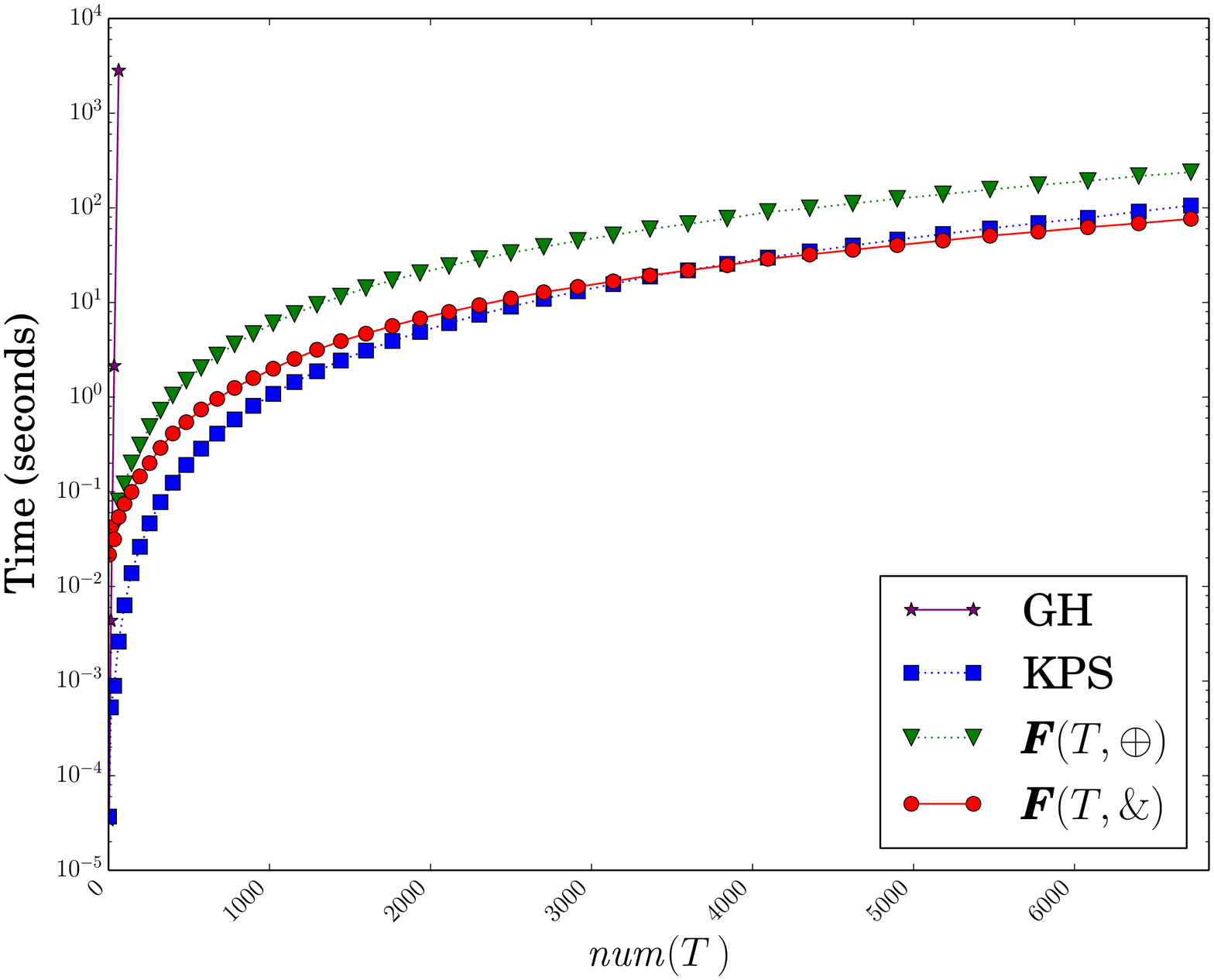}};
        \node[align=center, below = of recrcv,yshift=0cm]
    {\titlelabel{e}{Recursive-Receive (lin.\ scale)}};
      \end{tikzpicture}
    \caption{Benchmarks (2)}
  \label{fig:plots-2}
\end{figure}
 \subsection{Benchmark results}
The first set of benchmarks compares the performances
of the three approaches when the pair of types given are identical,
i.e., we measure the time it takes for an algorithm to check whether
$T \subtype T$ holds.
The second set of benchmarks considers types that are ``unfolded'', so
that types have different sizes.
Note that checking whether two equal types are in the subtyping
relation is one of the most costly cases of subtyping since every
branch of a choice must be visited.

Our results below show the performances of four algorithms:
($i$) our Haskell implementation of Gay and Hole's algorithm (GH),
($ii$) our implementation of Kozen, Palsberg, and Schwartzbach's
algorithm (KPS),
($iii$) an invocation to mCRL2 to check whether $U \models
\subform{T}$ holds, and
($iv$) an invocation to mCRL2 to check whether $T \models \SUPform{U}$
holds.

All the benchmarks were conducted on an 3.40GHz Intel i7 computer
with 16GB of RAM.
Unless specified otherwise, the tests have been executed with a
timeout set to $2$ hours ($7200$ seconds). A gap appears in the plots
whenever an algorithm reached the timeout.
Times ($y$-axis) are plotted on a \emph{logarithmic} scale, the
scale used for the size of types ($x$-axis) is specified below
each plot.

{\bf Arbitrary session types}
Plots (a) and (b) in Figure~\ref{fig:plots-1} shows how the algorithms
perform with arbitrary session types (randomly generated by our tool).
Plot (a) shows clearly that the execution time of KPS, $T \models
\SUPform{T}$, and $T \models \subform{T}$ mostly depends on
$\nummsg{T}$; while plot (b) shows that GH is mostly affected by the
number of messages in the unfolding of a type ($\unfold{T}$).

Unsurprisingly, GH performs better for smaller session types, but
starts reaching the timeout when $\nummsg{T} \approx 700$.
The other three algorithms have roughly similar performances, with the
model checking based ones performing slightly better for large session
types.
Note that both $T \models \SUPform{T}$ and $T \models \subform{T}$
have roughly the same execution time.

{\bf Non-recursive arbitrary session types}
Plot (c) in Figure~\ref{fig:plots-1} shows how the algorithms perform
with arbitrary types that do \emph{not} feature any recursive
definition (randomly generated by our tool), i.e.,
the types are of the form:
\[
\textstyle
T \coloneqq 
\tend
\bnfsep
\inchoice
\bnfsep
\outchoice
\]
The plot shows that GH performs much better than the other
three algorithms (terminating under $1$s for each invocation).
Indeed this set of benchmarks is the best case scenario for GH: there
is no recursion hence no need to unfold types.
Observe that the model checking based algorithms perform better than
KPS for large session types.
Again, $T \models \SUPform{T}$ and $T \models \subform{T}$ behave
similarly.

{\bf Handcrafted session types}
Plots (d) and (e) in Figure~\ref{fig:plots-2} shows how the algorithms
deal with ``super-recursive'' types, i.e., types of the form:
\[
\textstyle
T  \coloneqq  
\recND{x}_1 . \any_1 . \ldots \recND{x}_k . \any_k
\left\{
  \;
  \bigOp_{1 \leq i \leq k} 
  \any_i .
  \{ 
  \bigOp_{1 \leq j \leq k} \any_j . \var{x}_j
  \}
  \;
\right\}
\]
where $\nummsg{T} = k(k+2)$ for each $T$.
Plot (d) shows the results of experiments with $\bigtOp$ set to
$\inchoicetop$ and $\Op$ to $\sendop$; while $\bigtOp$ is set to
$\outchoicetop$ and $\Op$ to $\rcvop$ in plot (e).

The exponential time complexity of GH appears clearly in both plots:
GH starts reaching the timeout when  $\nummsg{T} = 80$ ($k=8$).
However, the other three algorithms deal well with larger session
types of this form. 
Interestingly, due to the nature of these session types (consisting of
either only \emph{internal} choices or only \emph{external} choices),
the two model checking based algorithms perform slightly differently.
This is explained by Definition~\ref{def:char-formula} where the
formula generated with $\SUPform{T}$ for an internal choice is larger
than for an external choice, and vice-versa for $\subform{T}$.
Observe that, $T \models \subform{T}$ (resp.\ $T \models \SUPform{T}$)
performs better than KPS for large session types in plot (d) (resp.\
plot (e)).

{\bf Unfolded types}
\newcommand{\unftype}{{V}}
The last set of benchmarks evaluates the performances of the
four algorithms to check whether 
$T = \rec{\varX} \unftype \; \subtype \; \rec{\varX} \left( \unftype
  \subs{\unftype}{\varX} \right) = U$
holds, where $\varX$ is fixed and $\unftype$ (randomly generated) is
of the form:
\[
\textstyle
\unftype \coloneqq 
\inchoiceSet{i}{I}{\unftype}
\bnfsep
\outchoiceSet{i}{I}{\unftype}
\bnfsep
\varX
\]
Plots (f) in Figure~\ref{fig:plots-2} shows the results of our
experiments (we have set the timeout to $6$ hours for these tests).
Observe that $U \models \subform{T}$ starts reaching the timeout quickly. In
this case, the model (i.e., $U$) is generally much larger than the
formula (i.e., $\subform{T}$).
After discussing with the mCRL2 team, this discrepancy seems to
originate from internal optimisations of the model checker that can be
diminished (or exacerbated) by tweaking the parameters of the
tool-set.
The other three algorithms have similar performances. Note that the
good performance of GH in this case can be explained by the fact that
there is only one recursion variable in these types; hence the size of
their unfolding does not grow very fast.

\section{Related work and conclusions}\label{sec:conc} \label{sec:related}
\paragraph{\bf Related work}
Subtyping for recursive types has been studied for many
years.
Amadio and Cardelli~\cite{AC93} introduced the first subtyping
algorithm for recursive types for the $\lambda$-calculus.
Kozen et al.\ gave a quadratic subtyping algorithm in~\cite{KPS95},
which we have adapted for session types, cf.\
Section~\ref{sub:kozen-et-al}.
A good introduction to the theory and history of the field 
is in~\cite{GLP02}.
Pierce and Sangiori~\cite{PS96} introduced subtyping for IO types in
the $\pi$-calculus, which later became a foundation for the algorithm
of Gay and Hole who first introduced subtyping for session
types in the $\pi$-calculus in~\cite{GH05}.
The paper~\cite{DemangeonH11} studied an abstract encoding between
linear types and session types, with a focus on subtyping.
Chen et al.~\cite{CDY2014} studied the notion of \emph{preciseness} of
subtyping relations for session types.
The present work is the first to study the algorithmic aspect of the
problem.

Characteristic formulae for finite processes were first studied
in~\cite{GS86}, then in~\cite{steffen89} for finite-state processes.
Since then the theory has been studied
extensively~\cite{si94,ai97,ai07,cs91,FS05,MOlm98,ails12} for most of
the van Glabbeek's spectrum~\cite{Glabbeek90} and in different
settings (e.g., time~\cite{AIPP00} and probabilistic~\cite{SZ12}).
See~\cite{ails12,ai07} for a detailed historical account of the field.
This is the first time characteristic formulae are applied to the
field of session types.
A recent work~\cite{ails12} proposes a general framework to obtain
characteristic formula constructions for simulation-like relation
``for free''.
We chose to follow~\cite{steffen89} as it was a better fit for session
types as they allow for a straightforward inductive construction of a
characteristic formula.
Moreover,~\cite{steffen89} uses the standard $\mu$-calculus which
allowed us to integrate our theory with an existing model checker.

\paragraph{\bf Conclusions}
In this paper, we gave a first connection between session types and
model checking, through a characteristic formulae approach based on
the $\mu$-calculus.
We gave three new algorithms for subtyping: two are based on model
checking and one is an instantiation of an algorithm for the
$\lambda$-calculus~\cite{KPS95}.
All of which have a quadratic complexity in the worst case and behave
well in practice.

Our approach can be easily: ($i$) adapted to types for the
$\lambda$-calculus (see Appendix~\ref{sec:lambda}) and ($ii$) extended
to session types that carry other (\emph{closed}) session types, e.g.,
see~\cite{GH05,CDY2014}, by simply applying the algorithm recursively
on the carried types. For instance, to check
$
\snd{a\langle \rcv{c} \outchoicetop \rcv{d} \rangle }
\; 
\subtype 
\;
\snd{a\langle \rcv{c} \rangle }   \, \inchoicetop \, \snd{b \langle \tend \rangle } 
$
one can check the subtyping for the outer-most types, while
building constraints, i.e., $\{ \rcv{\!c} \, \outchoicetop \rcv{d} \,
\subtype \, \rcv{c} \}$, to be checked later on, by re-applying the
algorithm.

The present work paves the way for new connections between session
types and modal fixpoint logic or model checking theories.
It is a basis for upcoming connections between model checking and
classical problems of session types, such as the asynchronous
subtyping of~\cite{CDY2014} and multiparty compatibility
checking~\cite{LTY15, DY13}.
We are also considering applying model checking approaches to session
types with probabilistic, logical~\cite{BHTY10}, or
time~\cite{BLY15,BYY14} annotations.
Finally, we remark that~\cite{CDY2014} also establishes that subtyping (cf.\
Definition~\ref{def:ab-subtype}) is \emph{sound} (but not complete)
wrt.\ the \emph{asynchronous} semantics of session types,
which models programs that communicate through FIFO buffers.
Thus, our new conditions (items $(b$)-$(e)$ of
Theorem~\ref{thm:safety-statements}) also imply safety $(a)$ in the
asynchronous setting.

\paragraph{\bf Acknowledgements}
We would like to thank Luca Aceto, Laura Bocchi, and Alceste Scalas
for their invaluable comments on earlier versions of this work.
This work is partially supported by UK EPSRC projects EP/K034413/1,
EP/K011715/1, and EP/L00058X/1; and by EU 7FP project under grant
agreement 612985 (UPSCALE).

\newpage
\bibliographystyle{abbrv}
\bibliography{model}

\begin{thebibliography}{10}

\bibitem{ai97}
L.~Aceto and A.~Ing{\'{o}}lfsd{\'{o}}ttir.
\newblock A characterization of finitary bisimulation.
\newblock {\em Inf. Process. Lett.}, 64(3):127--134, 1997.

\bibitem{ai07}
L.~Aceto and A.~Ing{\'{o}}lfsd{\'{o}}ttir.
\newblock Characteristic formulae: From automata to logic.
\newblock {\em Bulletin of the {EATCS}}, 91:58--75, 2007.

\bibitem{ails12}
L.~Aceto, A.~Ing{\'{o}}lfsd{\'{o}}ttir, P.~B. Levy, and J.~Sack.
\newblock Characteristic formulae for fixed-point semantics: a general
  framework.
\newblock {\em Mathematical Structures in Computer Science}, 22(2):125--173,
  2012.

\bibitem{AIPP00}
L.~Aceto, A.~Ing{\'{o}}lfsd{\'{o}}ttir, M.~L. Pedersen, and J.~Poulsen.
\newblock Characteristic formulae for timed automata.
\newblock {\em {ITA}}, 34(6):565--584, 2000.

\bibitem{AC93}
R.~M. Amadio and L.~Cardelli.
\newblock Subtyping recursive types.
\newblock {\em {ACM} Trans. Program. Lang. Syst.}, 15(4):575--631, 1993.

\bibitem{BHTY10}
L.~Bocchi, K.~Honda, E.~Tuosto, and N.~Yoshida.
\newblock A theory of design-by-contract for distributed multiparty
  interactions.
\newblock In {\em {CONCUR} 2010}, pages 162--176, 2010.

\bibitem{BLY15}
L.~Bocchi, J.~Lange, and N.~Yoshida.
\newblock Meeting deadlines together.
\newblock In {\em {CONCUR} 2015}, pages 283--296, 2015.

\bibitem{BYY14}
L.~Bocchi, W.~Yang, and N.~Yoshida.
\newblock Timed multiparty session types.
\newblock In {\em {CONCUR} 2014}, pages 419--434, 2014.

\bibitem{CDY2014}
T.-C. Chen, M.~Dezani-Ciancaglini, and N.~Yoshida.
\newblock On the preciseness of subtyping in session types.
\newblock In {\em PPDP 2014}, pages 146--135. ACM Press, 2014.

\bibitem{quickcheck}
K.~Claessen and J.~Hughes.
\newblock Quickcheck: a lightweight tool for random testing of {Haskell}
  programs.
\newblock In {\em {ICFP} 2000}, pages 268--279, 2000.

\bibitem{cs91}
R.~Cleaveland and B.~Steffen.
\newblock Computing behavioural relations, logically.
\newblock In {\em ICALP 1991}, pages 127--138, 1991.

\bibitem{zdlc}
Cognizant.
\newblock {Z}ero {D}eviation {L}ifecycle.
\newblock \url{http://www.zdlc.co}.

\bibitem{mcrl2}
S.~Cranen, J.~F. Groote, J.~J.~A. Keiren, F.~P.~M. Stappers, E.~P. de~Vink,
  W.~Wesselink, and T.~A.~C. Willemse.
\newblock An overview of the {mCRL2} toolset and its recent advances.
\newblock In {\em {TACAS} 2013}, pages 199--213, 2013.

\bibitem{DemangeonH11}
R.~Demangeon and K.~Honda.
\newblock Full abstraction in a subtyped pi-calculus with linear types.
\newblock In {\em {CONCUR} 2011}, pages 280--296, 2011.

\bibitem{DY13}
P.~Deni{\'{e}}lou and N.~Yoshida.
\newblock Multiparty compatibility in communicating automata: Characterisation
  and synthesis of global session types.
\newblock In {\em {ICALP} 2013}, pages 174--186, 2013.

\bibitem{haskell-smt}
I.~S. Diatchki.
\newblock Improving {Haskell} types with {SMT}.
\newblock In {\em Haskell 2015}, pages 1--10. ACM, 2015.

\bibitem{FS05}
H.~Fecher and M.~Steffen.
\newblock Characteristic mu-calculus formulas for underspecified transition
  systems.
\newblock {\em Electr. Notes Theor. Comput. Sci.}, 128(2):103--116, 2005.

\bibitem{GLP02}
V.~Gapeyev, M.~Y. Levin, and B.~C. Pierce.
\newblock Recursive subtyping revealed.
\newblock {\em J. Funct. Program.}, 12(6):511--548, 2002.

\bibitem{GH99}
S.~J. Gay and M.~Hole.
\newblock Types and subtypes for client-server interactions.
\newblock In {\em ESOP 2009}, pages 74--90, 1999.

\bibitem{GH05}
S.~J. Gay and M.~Hole.
\newblock Subtyping for session types in the pi calculus.
\newblock {\em Acta Inf.}, 42(2-3):191--225, 2005.

\bibitem{GS86}
S.~Graf and J.~Sifakis.
\newblock A modal characterization of observational congruence on finite terms
  of {CCS}.
\newblock {\em Information and Control}, 68(1-3):125--145, 1986.

\bibitem{groote2014modeling}
J.~F. Groote and M.~R. Mousavi.
\newblock {\em Modeling and analysis of communicating systems}.
\newblock MIT Press, 2014.

\bibitem{haskell-measures}
A.~Gundry.
\newblock A typechecker plugin for units of measure: Domain-specific constraint
  solving in {GHC} {Haskell}.
\newblock In {\em Haskell 2015}, pages 11--22. ACM, 2015.

\bibitem{HVK98}
K.~Honda, V.~T. Vasconcelos, and M.~Kubo.
\newblock Language primitives and type discipline for structured
  communication-based programming.
\newblock In {\em {ESOP} 1998}, pages 122--138, 1998.

\bibitem{betty-survey}
H.~H{\"u}ttel, I.~Lanese, V.~T. Vasconcelos, L.~Caires, M.~Carbone, P.-M.
  Deni{\'e}lou, D.~Mostrous, L.~Padovani, A.~Ravara, E.~Tuosto, et~al.
\newblock Foundations of behavioural types.
\newblock {\em Report of the EU COST Action IC1201 (BETTY)}, 2014.
\newblock \url{www.behavioural-types.eu/publications/WG1-State-of-the-Art.pdf}.

\bibitem{Kozen83}
D.~Kozen.
\newblock Results on the propositional mu-calculus.
\newblock {\em Theor. Comput. Sci.}, 27:333--354, 1983.

\bibitem{KPS95}
D.~Kozen, J.~Palsberg, and M.~I. Schwartzbach.
\newblock Efficient recursive subtyping.
\newblock {\em Mathematical Structures in Computer Science}, 5(1):113--125,
  1995.

\bibitem{tool}
J.~Lange.
\newblock Tool and benchmark data.
\newblock
  \url{http://bitbucket.org/julien-lange/modelcheckingsessiontypesubtyping},
  2015.

\bibitem{LTY15}
J.~Lange, E.~Tuosto, and N.~Yoshida.
\newblock From communicating machines to graphical choreographies.
\newblock In {\em {POPL} 2015}, pages 221--232, 2015.

\bibitem{appendix}
J.~Lange and N.~Yoshida.
\newblock Full version of this paper.
\newblock \url{www.doc.ic.ac.uk/~jlange/papers/char-formula-subtyping.pdf},
  2015.

\bibitem{Leino10}
K.~R.~M. Leino.
\newblock Dafny: An automatic program verifier for functional correctness.
\newblock In {\em {LPAR-16} 2010}, pages 348--370, 2010.

\bibitem{LY12}
K.~R.~M. Leino and K.~Yessenov.
\newblock Stepwise refinement of heap-manipulating code in {Chalice}.
\newblock {\em Formal Asp. Comput.}, 24(4-6):519--535, 2012.

\bibitem{MOlm98}
M.~M{\"{u}}ller{-}Olm.
\newblock Derivation of characteristic formulae.
\newblock {\em Electr. Notes Theor. Comput. Sci.}, 18:159--170, 1998.

\bibitem{piercebook02}
B.~C. Pierce.
\newblock {\em Types and Programming Languages}.
\newblock MIT Press, Cambridge, MA, USA, 2002.

\bibitem{PS96}
B.~C. Pierce and D.~Sangiorgi.
\newblock Typing and subtyping for mobile processes.
\newblock {\em Mathematical Structures in Computer Science}, 6(5):409--453,
  1996.

\bibitem{SZ12}
J.~Sack and L.~Zhang.
\newblock A general framework for probabilistic characterizing formulae.
\newblock In {\em {VMCAI} 2012}, pages 396--411, 2012.

\bibitem{scribble}
{{S}cribble {P}roject homepage}.
\newblock \url{www.scribble.org}.

\bibitem{steffen89}
B.~Steffen.
\newblock Characteristic formulae.
\newblock In {\em ICALP 1989}, pages 723--732, 1989.

\bibitem{si94}
B.~Steffen and A.~Ing{\'{o}}lfsd{\'{o}}ttir.
\newblock Characteristic formulae for processes with divergence.
\newblock {\em Inf. Comput.}, 110(1):149--163, 1994.

\bibitem{THK94}
K.~Takeuchi, K.~Honda, and M.~Kubo.
\newblock An interaction-based language and its typing system.
\newblock In {\em {PARLE} 1994}, pages 398--413, 1994.

\bibitem{Glabbeek90}
R.~J. van Glabbeek.
\newblock The linear time-branching time spectrum (extended abstract).
\newblock In {\em {CONCUR} 1990}, pages 278--297, 1990.

\bibitem{YHNN2013}
N.~Yoshida, R.~Hu, R.~Neykova, and N.~Ng.
\newblock The {Scribble} protocol language.
\newblock In {\em TGC 2013}, volume 8358, pages 22--41. Springer, 2013.

\end{thebibliography}

\appendix

\newpage
\section{Appendix: Proofs}\label{app:proofs}
\subsection{Compositionality} 
\lemcompo*
\begin{proof}\label{proof:lemcompo}
  By structural induction on the structure of $T$.
  \begin{enumerate}
  \item If $T = \tend$, then
    \begin{itemize}
    \item $T \subs{U}{\varX} = \tend$ and
      $\Anyform{T \subs{U}{\varX}}{\bigtOp} =   \mmbox{\Actions}{\falsek}$, and
    \item  $\Anyform{T}{\bigtOp} = 
      \mmbox{\Actions}{\falsek} = \mmbox{\Actions}{\falsek} \subs{ \Anyform{U}{\bigtOp} }{\varX}$.
    \end{itemize}
  \item If $T = \varX$, then 
    \begin{itemize}
    \item  $T \subs{U}{\varX} = U$, hence 
      $\Anyform{T \subs{U}{\varX}}{\bigtOp} = \Anyform{U}{\bigtOp}$, and
    \item $\Anyform{T}{\bigtOp} = \varX$, hence  
      $\Anyform{T}{\bigtOp} \subs{ \Anyform{U}{\bigtOp} }{\varX} = \Anyform{U}{\bigtOp}$.
    \end{itemize}
  \item If $T = \varY ( \neq \varX)$, then
    \begin{itemize}
    \item $\Anyform{\varY \subs{U}{\varX}}{\bigtOp} = \Anyform{\varY}{\bigtOp} = \varY$, and
    \item $\Anyform{\varY}{\bigtOp} \subs{ \Anyform{U}{\bigtOp} }{\varX}
      = \varY \subs{ \Anyform{U}{\bigtOp} }{\varX} = \varY$.
    \end{itemize}
  \item If $T = \choice$, then
    \begin{align*}
      \Anyform{T \subs{U}{\varX}}{\bigtOp} 
      & = \Anyform{ \choiceSetNoIdx{i}{I}{ T_i \subs{U}{\varX} }}{\bigtOp}
      \\
      & = \bigwedge_{i \in I} \mmdiamond{\opfun{\bigtOp}{a_i}}{\Anyform{T_i \subs{U}{\varX} }{\bigtOp}}
      \\
      \text{\tiny\it (I.H.)} \quad
      & =  \bigwedge_{i \in I} \mmdiamond{\opfun{\bigtOp}{a_i}}{ \left(\Anyform{T_i  }{\bigtOp}
          \subs{ \Anyform{U}{\bigtOp}}{\varX} \right)}
      \\
            & = 
      \left( 
        \bigwedge_{i \in I} \mmdiamond{\opfun{\bigtOp}{a_i}}{\Anyform{T_i  }{\bigtOp}}
      \right)
        \subs{ \Anyform{U}{\bigtOp}}{\varX} 
            \\
      &= \Anyform{T}{\bigtOp} \subs{ \Anyform{U}{\bigtOp} }{\varX}
    \end{align*}
          \item  If $T = \cochoice$, then
    \begin{align*}
      \Anyform{T \subs{U}{\varX}}{\bigtOp} 
      & = \Anyform{ \cochoiceSetNoIdx{i}{I}{ T_i \subs{U}{\varX} }}{\bigtOp}
      \\
      & = 
      \bigwedge_{i \in I} \mmbox{\opfun{\bigtOp}{a_i}}{\Anyform{  T_i \subs{U}{\varX} }{\bigtOp}}
      \, \land \,
      \\
      & \qquad 
      \bigvee_{i \in I} \mmdiamond{\opfun{\bigtOp}{a_i}}{\truek}
      \, \land \,
      \mmbox{ \compset{ \{ \opfun{\bigtOp}{a_i} \st i \in I\} }}{\falsek}
      \\
      \text{\tiny\it (I.H.)} \quad
      & = 
      \bigwedge_{i \in I} \mmbox{\opfun{\bigtOp}{a_i}}{
        \left(  \Anyform{  T_i  }{\bigtOp} \subs{\Anyform{U}{\bigtOp}}{\varX} \right)
      }
      \, \land \,
      \\
      & \qquad 
      \bigvee_{i \in I} \mmdiamond{\opfun{\bigtOp}{a_i}}{\truek}
      \, \land \,
      \mmbox{ \compset{ \{ \opfun{\bigtOp}{a_i} \st i \in I\} }}{\falsek}
      \\
      & =  \left( \bigwedge_{i \in I} \mmbox{\opfun{\bigtOp}{a_i}}{
          \Anyform{  T_i  }{\bigtOp}
        }
        \, \land \,
      \right.
      \\
      & \qquad 
      \left.
        \bigvee_{i \in I} \mmdiamond{\opfun{\bigtOp}{a_i}}{\truek}
        \, \land \,
        \mmbox{ \compset{ \{ \opfun{\bigtOp}{a_i} \st i \in I\} }}{\falsek}
      \right) \subs{\Anyform{U}{\bigtOp}}{\varX} 
      \\
      & = \Anyform{T}{\bigtOp} \subs{ \Anyform{U}{\bigtOp} }{\varX}
    \end{align*}
          \item If $ T = \rec{\varY} T'$ ($\varX \neq \varY$), we have
    \begin{align*}
      \Anyform{ \rec{\varY} T' \subs{U}{\varX}}{\bigtOp} 
      & =  \mmnu{\varY} \Anyform{T'  \subs{U}{\varX} }{\bigtOp}
      \\
      \text{\tiny\it (I.H.)} \quad
      & = \mmnu{\varY} \Anyform{T'}{\bigtOp}  \subs{\Anyform{U}{\bigtOp}}{\varX}
      \\ 
      & = \Anyform{T}{\bigtOp} \subs{ \Anyform{U}{\bigtOp} }{\varX}
    \end{align*}
  \end{enumerate}
\end{proof}

\subsection{Extensions and approximations}
The proofs in this section follow closely the proof techniques
in~\cite{steffen89}.
\begin{mydef}[Extended subtyping]
  Let $T, U \in \TYPES$, $\fora \in \FORMULA$, and
  $\vec{\varX} = (\varX_1 , \ldots , \varX_n)$ be
  a vector containing all the free variables in $T$, $U$, or $\fora$.
    We define the \emph{extended subtyping} $\subtype_e$
  and the \emph{extended satisfaction relation}, $\models_e$, by
  \begin{enumerate}
  \item $T \subtype_e U \iff \forall \vec{V} \in \TYPES^n \qst 
    T\subs{\vec{V}}{\vec{\varX}} \subtype 
    U\subs{\vec{V}}{\vec{\varX}} $
  \item $ T \models_e \fora
    \iff 
    \forall  \vec{V} \in \TYPES^n 
    \forall \vec{\forb} \in \FORMULA^n 
    \qst
    \vec{V} \models \vec{\forb}
    \implies
    T\subs{\vec{V}}{\vec{\varX}}
    \models \phi \subs{\vec{\forb}}{\vec{\varX}}
    $
  \end{enumerate}
  where $  \vec{V} \models \vec{\forb}$ is understood component wise.
  \end{mydef}

\begin{mydef}[Subtyping approximations]\label{def:subtype-approx}
  Let $T, U \in \TYPES$ and
  $\vec{\varX} = (\varX_1 , \ldots , \varX_n)$ be
  a vector containing all the free variables in $T$ or $U$.
    The extended $k$-limited subtyping, $\subtype_{e,k}$ is defined
  inductively on $k$ as follows: $T \subtype_{e,0} U$ always holds; if $k
  \geq 1$, then $T \subtype_{e,k} U$ holds iff for all $\vec{V} \in
  \TYPES^n$, $ T\subs{\vec{V}}{\vec{\varX}} \subtype_{e,k}
  U\subs{\vec{V}}{\vec{\varX}} $ can be derived from the following
  rules:
  \[
  \begin{array}{c}
        \inference{S-out}
    {
      I \subseteq J
      &
      \forall i \in I \qst T_i \absub{\bigOp}_{e,k-1} U_i
    }
    {
      \choice \absub{\bigOp}_{e,k} \choiceSet{j}{J}{U}
    }
        \qquad
    \inference{S-in}
    {
      J \subseteq I
      &
      \forall j \in J \qst T_j \absub{\bigOp}_{e,k-1} U_j
    }
    {
      \cochoice \absub{\bigOp}_{e,k}   \cochoiceSet{j}{J}{U}
    }
        \\[1pc]
    \inference{S-end}
    {}
    {\tend \absub{\bigOp}_{e,k} \tend}
      \end{array}
  \]
    Recall that we are assuming an equi-recursive view of types.
\end{mydef}

\begin{restatable}{lemma}{lemextsubtype}
  \label{lem:ext-subtype}
  $ T \absub{\bigOp}_e U \iff \forall k \qst T \absub{\bigOp}_{e,k} U$
\end{restatable}
\begin{proof}
  The ($\Rightarrow$) direction is straightforward, while the converse
  follow from the fact the session types we consider have only a
  finite number of states.
  \end{proof}

\begin{mydef}[Semantics approximations]
  Let $T \in \TYPES$ and $\vec{\varX} = (\varX_1 , \ldots , \varX_n)$
  be a vector containing all the free variables in $T$.
    The extended $k$-limited satisfaction relation $\models_{e,k}$ is
  defined inductively as follows on $k$:
    $T \models_{e,0} \fora$ always holds;
  if $k \geq 1$, then $\models_{e,k}$ is given by:
    \[
    \begin{array}{l@{\quad \mathit{iff} \quad}l}
      \multicolumn{1}{l}{T \ekmodels \truek}
      \\
      T \ekmodels \fora_1 \land \fora_2 & T \ekmodels \fora_1 \text{ and } T \ekmodels \fora_2
      \\
      T \ekmodels \fora_1 \lor \fora_2 & T \ekmodels \fora_1 \text{ or } T \ekmodels \fora_2
      \\
      T \ekmodels \mmbox{\any}{\fora}
      & 
      \forall \vec{V} \in \TYPES^n \;
      \forall T' \qst \text{if } 
      T\subs{\vec{V}}{\vec{\varX}}  \semarrow{\any} T' \text{ then } T' \models_{e,k-1} \fora
      \\
      T \ekmodels \mmdiamond{\any}{\fora}
      & 
      \forall \vec{V} \in \TYPES^n \;
      \exists T' \qst
      T\subs{\vec{V}}{\vec{\varX}}  \semarrow{\any} T' \text{ and } T' \models_{e,k-1} \fora
      \\
      T \ekmodels \mmnu{\varX} \fora 
      &
      \forall n \qst T \ekmodels \approxi{ \mmnu{\varX} \fora }{n}
      \ensuremath{$\qedhere$}
    \end{array}
    \]
    \end{mydef}

\begin{restatable}{lemma}{lemextmodels}
  \label{lem:ext-models}
  $ T \models_e \fora \iff \forall k \geq 0 \qst T \ekmodels \fora$
\end{restatable}
\begin{proof}
  The $(\Rightarrow)$ direction is straightforward, while the
  $(\Leftarrow)$ direction follows from the fact that a session type induce a
  finite LTS.
  \end{proof}

\begin{restatable}[Fixpoint properties]{lemma}{lemfixpointprop}
  \label{lem:fix-point-prop}
  Let $T \in \TYPES$ and $\fora \in \FORMULA$, then we have:
  \begin{enumerate}
  \item 
    \label{enum:fix-mu}
    $T \ekmodels \mmnu{\varX} \fora \iff T \ekmodels \phi \subs{\mmnu{\varX} \fora}{\varX}$
  \item 
    \label{enum:fix-rec}
    $\rec{\varX} T \ekmodels \fora \iff T\subs{\rec{\varX}T}{\varX} \ekmodels \phi$
  \item 
    \label{enum:fix-tr}
    $ 
    \rec{\varX} T 
    \; \absub{\bigOp}_{e,k} \;
    T \subs{\rec{\varX} T}{\varX}
    \; \absub{\bigOp}_{e,k} \;
    \rec{\varX} T 
    $
  \end{enumerate}
\end{restatable}
\begin{proof}
  The first property is a direct consequence of the definition of
  $\models_{e,k}$, while the last two properties follow from the
  equi-recursive view of types.
\end{proof}

\subsection{Main results}

\thmmaintheorem*
\begin{proof}\label{proof:thmmaintheorem}
  Direct consequence of Lemma~\ref{lem:main-lemma}.
\end{proof}

\begin{restatable}[Main lemma]{lemma}{lemmainlemma}
  \label{lem:main-lemma}
    $\forall T, U \in \TYPES \qst 
  T \absub{\bigOp}_{e} U
  \iff
  U \models_e \Anyform{T}{\bigtOp}$
\end{restatable}
\begin{proof}
  According to Lemmas~\ref{lem:ext-subtype} and~\ref{lem:ext-models},
  it is enough to show that
  \begin{equation}\label{eq:main-equi}
    \forall k \geq 0 
    \qst
    \forall U, T \in \TYPES
        \qst
    T \absub{\bigOp}_{e,k} U 
    \; \iff  \;
    U \ekmodels \Anyform{T}{\bigtOp}
  \end{equation}
  We show this by induction on $k$.
  If $k = 0$, the result holds trivially, let us show that it also holds for $k \geq 1$.
  We distinguish four cases according to the structure of $T$.
  \begin{enumerate}
  \item If $T = \varX$, then must have $U = \varX$, by definition of
  $\absub{\bigOp}_e$ and $\models_e$.
    \item If $T = \rec{\varX} T'$, then by Lemma~\ref{lem:fix-point-prop}, we have
    \begin{enumerate}
    \item 
      $ 
      U \ekmodels  \Anyform{T}{\bigtOp} 
      \iff
      U \ekmodels  \Anyform{T'}{\bigtOp}  \subs{\Anyform{T}{\bigtOp}}{\varX}
      $
    \item
      $ 
      T 
      \; \absub{\bigtOp}_{e,k} \;
      T' \subs{T}{\varX}
      \; \absub{\bigtOp}_{e,k} \;
      T 
      $
    \end{enumerate}
        Applying Lemma~\ref{lem:compo}, it is enough to show that:
    \[
    \forall T, U \in \TYPES 
    \qst
    T' \subs{\rec{\varX} T'}{\varX}
    \absub{\bigtOp}_{e,k}
    U
    \; \iff \;
    U \ekmodels \Anyform{T' \subs{\rec{\varX} T'}{\varX}}{\bigtOp}
    \]

        Hence, since we have assumed that the types are guarded, 
    we only have to deal with the cases where
    $T = \choice$, $T = \cochoice$, and $T = \tend$.
        
    On the other hand, considering both sides of the
    equivalence~\eqref{eq:main-equi}, we notice that $U$ cannot
    be a variable. Thus, let us assume that
    $U = \rec{\varX} U'$, by Lemma~\ref{lem:fix-point-prop}, we have
       \begin{enumerate}
    \item 
      $ 
      U \ekmodels  \Anyform{T}{\bigtOp} 
      \iff
      U' \subs{U}{\varX}
      \ekmodels  \Anyform{T}{\bigtOp}
      $
    \item
      $ 
      U
      \; \absub{\bigtOp}_{e,k} \;
      U' \subs{U}{\varX}
      \; \absub{\bigtOp}_{e,k} \;
      U 
      $
    \end{enumerate}
    Hence, applying Lemma~\ref{lem:compo} again, this case reduces to
    the cases where $U$ is of the form: $\choiceSet{j}{J}{U}$,
    $\cochoiceSet{j}{J}{U}$, or $\tend$.
          \item  $T = \tend$
    \begin{itemize}
    \item $(\Rightarrow)$ Assume $\tend = T \absub{\bigtOp}_{e,k} U$,
      then by Definition~\ref{def:subtype-approx}, we have $U = \tend$.
            By Definition~\ref{def:char-formula}, we have
      $\Anyform{\tend}{\bigtOp} = \mmbox{\Actions}{\falsek}$,
      and we 
      have $\tend \ekmodels  \mmbox{\Actions}{\falsek}$ since $ U = \tend \nsemarrow$.
          \item $(\Leftarrow)$ Assume $U \ekmodels \Anyform{\tend}{\bigtOp}$.
      By Definition~\ref{def:char-formula}, we have
      $U \ekmodels  \mmbox{\Actions}{\falsek}$,
            which holds iff $U \nsemarrow$, hence we must have $U = \tend$.
            Finally, by Definition~\ref{def:subtype-approx}, we have $\tend
      \ekabsub \tend$.
    \end{itemize}
      \item $T = \choice$
        \begin{itemize}
    \item $(\Rightarrow)$ Assume $\choice \absub{\bigtOp}_{e,k} U$.
            By  Definition~\ref{def:subtype-approx}, $ U = \choiceSet{j}{J}{U}$
      with $I \subseteq J$ (note that $\emptyset \neq I$ by assumption)
      and $\forall i \in I \qst T_i \absub{\bigtOp}_{e,k-1} U_i$.
            Hence, $\forall i \in I \qst U \semarrow{\anydir{a_i}} U_i$, and
      by induction hypothesis, we have $U_i \models_{e,k-1} \Anyform{T_i}{\bigtOp}$, for
      all $i \in I$.
      
      By Definition~\ref{def:char-formula}, we have
      $
      \Anyform{T}{\bigtOp} = \bigwedge_{i \in I} \mmdiamond{\opfun{\bigtOp}{a_i}}{\Anyform{T_i}{\bigtOp}}
      $.
      Thus we have to show that for all $i \in I$, $U \semarrow{\opfun{\bigtOp}{a_i}} U_i$
      and $U_i \models_{e,k-1} \Anyform{T_i}{\bigtOp}$; which follows from above.

          \item $(\Leftarrow)$ Assume $U \ekmodels \Anyform{\choice}{\bigtOp}$.
      From Definition~\ref{def:char-formula}, we have
      \[
      \Anyform{T}{\bigtOp} = \bigwedge_{i \in I} \mmdiamond{\opfun{\bigtOp}{a_i}}{\Anyform{T_i}{\bigtOp}}
      \]
            Hence, , $\forall i \in I \qst U \semarrow{\opfun{\bigtOp}{a_i}} U_i$, and 
      $U_i \models_{e,k-1} \Anyform{T_i}{\bigtOp}$, for all $i \in I$.
            Hence, we must have $U  =   \choiceSet{j}{J}{U}$ with $I \subseteq J$ and
            by induction hypothesis, this implies that $T_i \absub{\bigtOp}_{e,k-1} U_i$ for all $i \in I$.
                \end{itemize}
  \item $T = \cochoice$
    \begin{itemize}
    \item $(\Rightarrow)$ Assume $\cochoice \absub{\bigtOp}_{e,k} U$.
      By  Definition~\ref{def:subtype-approx}, $ U = \cochoiceSet{j}{J}{U}$, with
      $J \subseteq I$ and $\forall j \in J \qst T_j \absub{\bigtOp}_{e,k-1} U_j$.
            Hence, by induction hypothesis, we have $U_j \models_{e,k-1}
      \Anyform{T_j}{\bigtOp}$, for all $j \in J$.

      By  Definition~\ref{def:char-formula}, we have
            \begin{equation}\label{eq:formula-cochoice-1}
        {\Anyform{T}{\bigtOp}} \;  =  \;
        \bigwedge_{i \in I} \mmbox{\opfun{\bigtOp}{a_i}}{\Anyform{T_i}{\bigtOp}}
        \; \land \;
        \bigvee_{i \in I} \mmdiamond{\opfun{\bigtOp}{a_i}}{\truek}
        \; \land \;
        \mmbox{ \compset{ \{ \opfun{\bigtOp}{a_i} \st i \in I\} }}{\falsek}
      \end{equation}
            We must show that $U \ekmodels  {\Anyform{T}{\bigtOp}} $.
            Since $J \subseteq I$, we have that $\forall i
      \in I \qst T \semarrow{\opfun{\bigtOp}{a_i}} T_i \implies U
      \semarrow{\opfun{\bigtOp}{a_i}} U_i$, hence the first conjunct
      of~\eqref{eq:formula-cochoice-1} holds (using the induction
      hypothesis, cf.\ above).
            While the second conjunct of~\eqref{eq:formula-cochoice-1} must
      be true from the assumption that $ \emptyset \neq J$.
            Finally, the third conjunct of~\eqref{eq:formula-cochoice-1} is
      false only if $U \semarrow{\opfun{\bigtOp}{a_n}}$ with $n \notin I$,
      which contradicts $J \subseteq I$.
          \item $(\Leftarrow)$ Assume $U \ekmodels \Anyform{\cochoice}{\bigtOp}$.
      From  Definition~\ref{def:char-formula}, we have
            \begin{equation}\label{eq:formula-cochoice-2}
        {\Anyform{T}{\bigtOp}} \;  =  \;
        \bigwedge_{i \in I} \mmbox{\opfun{\bigtOp}{a_i}}{\Anyform{T_i}{\bigtOp}}
        \; \land \;
        \bigvee_{i \in I} \mmdiamond{\opfun{\bigtOp}{a_i}}{\truek}
        \; \land \;
        \mmbox{ \compset{ \{ \opfun{\bigtOp}{a_i} \st i \in I\} }}{\falsek}
      \end{equation}
      Hence, we must have $U = \cochoiceSet{j}{J}{U}$. It follows
      straightforwardly that $\emptyset \neq J \subseteq I$.
            Finally, the fact that for all $j \in J \qst T_j
      \absub{\bigtOp}_{e,k-1} U_j$, follows from the induction
      hypothesis.
            \qedhere
    \end{itemize}
      \end{enumerate}
\end{proof}

 \subsection{Duality and safety in session types}\label{app:safety}
\thmcharformduality*
\begin{proof} \label{proof:thmcharformduality}
  By straightforward induction on the structure of $T$. 
      \begin{enumerate}
  \item The result follows trivially if $T = \tend$ or $T = \varX$.
      \item If $T = \rec{\varX} T'$, then we have $\dual{\Anyform{\rec{\varX} T'}{\bigtOp}} =
    \mmnu{\varX} \dual{ \Anyform{T'}{\bigtOp}}$,
    and
    $
    \Anyform{\dual{T}}{\dual{\bigtOp}} = 
    \mmnu{\varX} \Anyform{\dual{T'}}{\dual{\bigtOp}}
    $. The result follows by induction hypothesis.
      \item If $T = \choice$, then we have 
        \begin{align*}
      \dual{\Anyform{T}{\bigtOp}}
      & = \dual { \bigwedge_{i \in I} \mmdiamond{\opfun{\bigtOp}{a_i}}{\Anyform{T_i}{\bigtOp}} }
      \\
      & =  \bigwedge_{i \in I} \mmdiamond{\opfun{\bigcoOp}{a_i}}{ \dual { \Anyform{T_i}{\bigtOp} } }
      \\
      \text{\tiny\it (I.H.)} \quad
      & = \bigwedge_{i \in I} \mmdiamond{\opfun{\bigcoOp}{a_i}}{  \Anyform{\dual{T_i}}{\dual{\bigtOp}} } 
            =  \Anyform{\dual{T}}{\dual{\bigtOp}}
    \end{align*}
      \item If $T = \cochoice$, then we have 
    \begin{align*}
      \dual{\Anyform{T}{\bigtOp}}
      & =   \dual{
        \bigwedge_{i \in I} \mmbox{\opfun{\bigtOp}{a_i}}{ { \Anyform{T_i}{\bigtOp}} }
      \; \land \;
      \bigvee_{i \in I} \mmdiamond{\opfun{\bigtOp}{a_i}}{\truek}
      \; \land \;
      \mmbox{ \compset{ \{ \opfun{\bigtOp}{a_i} \st i \in I\} }}{\falsek}
      }
      \\
      & = \bigwedge_{i \in I} \mmbox{\opfun{\bigcoOp}{a_i}}{ \dual{ \Anyform{T_i}{\bigtOp}} }
      \; \land \;
      \bigvee_{i \in I} \mmdiamond{\opfun{\bigcoOp}{a_i}}{\truek}
      \; \land \;
      \mmbox{ \compset{ \{ \opfun{\bigcoOp}{a_i} \st i \in I\} }}{\falsek}
      \\
      \text{\tiny\it (I.H.)} \quad
      & = \bigwedge_{i \in I} \mmbox{\opfun{\bigcoOp}{a_i}}{ \Anyform{ \dual{T_i}}{ \dual{\bigtOp}} }
      \; \land \;
      \bigvee_{i \in I} \mmdiamond{\opfun{\bigcoOp}{a_i}}{\truek}
      \; \land \;
      \mmbox{ \compset{ \{ \opfun{\bigcoOp}{a_i} \st i \in I\} }}{\falsek}
      \\
      & = \Anyform{\dual{T}}{\dual{\bigtOp}}
    \end{align*}
      \end{enumerate}
  \end{proof}

\thmsafety*
\begin{proof}
  $(\Longleftarrow)$ We prove that if $T \subtype \dual{U}$ then 
$T \spar U$ is safe by coinduction on the derivation of $T \subtype \dual{U}$
(recall that  $\subtype$ stands for $\absub{\,\inchoicetop}$). \\[1mm]
{\bf Case [S-end]} Obvious since $T =\dual{U}=\tend$ and 
$T \spar U \not \semarrow{}$. \\[1mm] 
{\bf Case [S-$\bigtOp$]} Suppose $T = \inchoice$. Then
$\dual{U}=\inchoiceSet{j}{J}{\dual{U}}$ such that $I \subseteq J$ and
$T_i \subtype \dual{U}_i$ for all $i \in I$. For all $a_i$ such that
$i \in I$, $T \semarrow{!a_i} T_i$ implies $U \semarrow{?a_i} U_i$.
Hence by [S-COM], we have $T \spar U\semarrow{} T_i \spar U_i$. Then
by coinduction hypothesis,
$T_i  \spar U_i$ is safe. \\[1mm]
{\bf Case [S-$\dual{\bigtOp}$]} Similar to the above case.\\[1mm]
$(\Longrightarrow)$ We prove 
$(\neg (T \subtype \dual{U}) \land \neg (\dual{T} \subtype U))$ implies 
$T \spar U$ has an error. Since the error rule coincides 
with the negation rules of subtyping in~\cite[Table 7]{CDY2014}, 
we conclude this direction. 
\end{proof}

\newpage
\newcommand{\bott}{\mathtt{bot}}
\newcommand{\topp}{\mathtt{top}}
\newcommand{\tsemarrow}[1]{\xhookrightarrow{\, #1 \,}}
\newcommand{\dirr}{\delta}
\newcommand{\lamsub}[2]{{{\Lambda}}(#1,#2)}
\newcommand{\RTYPES}{\mathcal{T}_R}
\newcommand{\trec}[1]{\mathtt{rec}\, {#1} . }
\newcommand{\tmmnu}[1]{\nu {#1} .\,}

\section{Appendix: Recursive types for the $\lambda$-calculus}
\label{sec:lambda}

\subsection{Recursive types and subtyping}
We consider recursive types for the $\lambda$-calculus below:
\[
t \coloneqq
\topp 
\bnfsep 
\bott 
\bnfsep
t_0 \rightarrow t_1
\bnfsep
\trec{v} t
\bnfsep
v
\]
Let $\RTYPES$ be the set of all closed recursive types.

A type $t$ induces an LTS according to the rules below:
\[
\begin{array}{c@{\qquad\qquad}c}
  \inference{top}
  {}
  {\topp \tsemarrow{\topp} \topp }
  &
  \inference{bot}
  {}
  {\bott \tsemarrow{\bott} \topp}
  \\[0.5cm]
  \inference{arrow}
  {i \in \{0,1\}}
  {t_0 \rightarrow t_1 \; \tsemarrow{i} \; t_i }
  &
  \inference
  {rec}
  {t \subs{\trec{v} t}{v} \tsemarrow{a} t'}
  {\trec{v} t \tsemarrow{a} t'}
\end{array}
\]
where we let $a$ range over $\{ 0, 1, \bott, \topp \}$.

\begin{definition}[Subtyping for recursive types]
  $\subtype \subseteq \RTYPES \times \RTYPES$ is the
  \emph{largest} relation that contains the rules:
  \[
  \coinference{S-bot}
  {t \in \RTYPES}
  {\bott \subtype t} 
  \qquad
  \coinference{S-top}
  {t \in \RTYPES}
  {t \subtype \topp} 
            \qquad
  \coinference{S-arrow}
  {
    t'_0 \subtype t_0
    &
    t_1 \subtype t'_1
  }
  {t_0 \rightarrow t_1  \subtype  t'_0 \rightarrow t'_1}
  \]
  Recall that we are
  assuming an equi-recursive view of types.
  The double line in the rules indicates that the rules
  should be interpreted \emph{coinductively}.
\end{definition}

\subsection{Characteristic formulae for recursive types}

We assume the same fragment of the modal $\mu$-calculus as in
Section~\ref{sub:mucal} but for ($i$) omitting the direction $\Op$ on
labels, i.e., we consider modalities: $\mmbox{a}{\fora}$ and
$\mmdiamond{a}{\fora}$; and ($ii$) using $v$ to range over recursion
variables.

Let $\dirr \in \{ \topp , \bott \}$, and $\dual{\bott} = \topp$,
$\dual{\topp} = \bott$.

\[
\lamsub{t}{\dirr} \defi
\begin{cases}
  \mmdiamond{\dirr}\truek
  &
  \text{if } t = \dirr
  \\
  \truek
  &
  \text{if } t = \dual{\dirr}
  \\
  \mmdiamond{0} \, \lamsub{t_0}{ \dual{\dirr} }
  \; \land \;
  \mmdiamond{1} \, \lamsub{t_1}{\dirr}
  &
  \text{if } t = t_0 \rightarrow t_1
  \\
  \tmmnu{v}  \lamsub{t'}{\dirr}
  & 
  \text{if } t = \trec{v} t'
  \\
  v
  & 
  \text{if } t = v
\end{cases}
\]

\begin{theorem} The following holds:
  \begin{itemize}
  \item $t \subtype t' \iff t' \models \lamsub{t}{\, \topp}$
  \item $t \subtype t' \iff t \; \models \lamsub{t'}{\, \bott}$
  \end{itemize}
\end{theorem}
\begin{proof}
  We show only the $(\Leftarrow)$ direction here.
    \begin{enumerate}
  \item \label{en:lambda-top}
    We show $t \subtype t' \Leftarrow t' \models \lamsub{t}{\topp}$
    by induction on $t$.
        \begin{itemize}
    \item If $t = \topp$, then $\lamsub{t}{\topp} =
      \mmdiamond{\topp}\truek$, hence $t' = \topp$.
          \item If $t = \bott$, then $\lamsub{t}{\topp} = \truek$ hence $t'$
      can be any type, as expected.
          \item If $t = t_0 \rightarrow t_1$, then
      \[
      \lamsub{t}{\topp} = 
      \mmdiamond{0} \, \lamsub{t_0}{ \bott }
      \; \land \;
      \mmdiamond{1} \, \lamsub{t_1}{\topp}
      \]
      hence we must have $t' = t'_0 \rightarrow t'_1$  with
      $t'_0 \models  \lamsub{t_0}{ \bott }$
      (hence $t'_0 \subtype t_0$, by IH, see below)
      and
      $t'_1 \models \lamsub{t_1}{\topp}$
            (hence $t_1 \subtype t'_1$ by IH).
    \end{itemize}
  \item \label{en:lambda-bot}
    We show $t \subtype t' \Leftarrow t \models \lamsub{t'}{\bott}$
    by induction on $t'$.
    \begin{itemize}
    \item If $t' = \bott$, then  $\lamsub{t'}{\bott} =  \mmdiamond{\bott}\truek$ 
      and $t = \bott$.
          \item if $t' = \topp$, then $\lamsub{t'}{\bott} = \truek$ and $t$
      can be any type, as expected.
          \item If $t' = t'_0 \rightarrow t'_1$, then
      \[
      \lamsub{t'}{\bott} = 
      \mmdiamond{0} \, \lamsub{t'_0}{ \topp }
      \; \land \;
      \mmdiamond{1} \, \lamsub{t'_1}{\bott}
      \]
      hence we must have $t = t_0 \rightarrow t_1$  with
      $t_0 \models  \lamsub{t'_0}{ \topp }$
      (hence $t'_0 \subtype t_0$, by IH, see above)
      and
      $t_1 \models \lamsub{t'_1}{\bott}$
            (hence $t_1 \subtype t'_1$ by IH).
    \end{itemize}
  \end{enumerate}
  The other direction is similar to the above, while the recursive
  step is similar to the proof of Theorem~\ref{thm:main-theorem}.
\end{proof}

\end{document}